\documentclass[onecolumn,pra,longbibliography,nofootinbib]{revtex4-2}

\usepackage[dvips]{graphicx} 
\usepackage{amsfonts,amscd,amsmath,amsthm,amssymb}
\allowdisplaybreaks[4]
\usepackage{enumerate}
\usepackage{epsfig}
\usepackage{caption}
\usepackage{subcaption}
\usepackage{xcolor}
\usepackage[colorlinks = true]{hyperref}
\usepackage{physics}
\usepackage{epstopdf}
\usepackage{framed}
\usepackage{threeparttable,booktabs,multirow}
\usepackage{color}
\usepackage{comment}
\usepackage[ruled,vlined]{algorithm2e}
\usepackage[most]{tcolorbox}
\usepackage{tikz}

\graphicspath{{./figure/}}

\usepackage{tikz}
\usetikzlibrary{tikzmark, calc, fit, positioning}
\usetikzlibrary{shapes,arrows,arrows.meta}
\usetikzlibrary{quantikz2}

\newtheorem{theorem}{Theorem}
\newtheorem{lemma}{Lemma}
\newtheorem{corollary}{Corollary}

\newtheorem{definition}{Definition}

\newtcolorbox[auto counter]{mybox}[2][]{
	enhanced,
	breakable,
	colback=blue!5!white,
	colframe=blue!75!black,
	fonttitle=\bfseries,
	title=Box \thetcbcounter: #2,#1
}

\newcommand{\CNOT}{\mathrm{CNOT}}

\begin{document}
\preprint{APS/123-QED}

\title{Constant-depth adaptive preparation of Dicke and symmetric states}
\author{Rui Luo}
\affiliation{Center for Quantum Information, Institute for Interdisciplinary Information Sciences, Tsinghua University, Beijing, 100084 China}
\author{Junjie Chen}
\affiliation{Center for Quantum Information, Institute for Interdisciplinary Information Sciences, Tsinghua University, Beijing, 100084 China}
\author{Xiongfeng Ma}
\email{xma@tsinghua.edu.cn}
\affiliation{Center for Quantum Information, Institute for Interdisciplinary Information Sciences, Tsinghua University, Beijing, 100084 China}

\begin{abstract}

Efficient preparation of Dicke states and, more generally, permutation-symmetric states is important for quantum metrology, quantum networking, and collective quantum information processing. Measurements and classical feedforward enable low-depth preparations of these states, with a cost of ancillary qubits. In this work, we introduce an exact constant-depth adaptive preparation protocol for arbitrary Dicke-$(n,k)$ states and further symmetric states. We first provide a protocol preparing the uniform subset superposition state, as a primitive, using constant-depth adaptive circuit with $O(k^2\log^2 n)$ ancillary qubits and success probability at least $1/k$. This yields an exact, probabilistic, constant-depth Dicke-state preparation protocol using $O\left(n^2+k^2\log^2 n+kn\log n\log\log n\right)$ ancillary qubits. Parallel repetition suppresses the failure probability exponentially without increasing the quantum depth. Moreover, the uniform subset superposition state is also of independent interest as the uniform vertex state of the Johnson graph and as the compact uniform subset state appearing in quantum-walk and topological-data-analysis algorithms. We further establish a general lifting framework that coherently combines clean unitary Dicke-state preparation circuits to prepare arbitrary symmetric states with only polynomial ancillary overhead. Combined with recent constant-depth unitary Dicke-state constructions, this gives an exact constant-depth preparation protocol for arbitrary $n$-qubit symmetric states using $O(n^3\sqrt{\log n})$ ancillary qubits.
\end{abstract}

\maketitle


\section{Introduction}

Efficient quantum-state preparation is a fundamental primitive in quantum computing and quantum information processing. The ability to prepare high-fidelity entangled states underlies many tasks, including quantum simulation~\cite{blatt2012quantum,Bloch2008Manybody}, fault-tolerant quantum computation~\cite{shor1995scheme,gottesman1997stabilizer,nielsen2010quantum}, measurement-based quantum computing~\cite{raussendorf2003measurement,briegel2009measurement}, and quantum communication protocols~\cite{gisin2007quantum,Cleve1999How,Lu2016Secret}. As quantum devices continue to scale, reducing circuit depth while keeping other resource overheads manageable becomes increasingly important, since shallow circuits shorten the coherent execution time and limit the accumulation of noise.

Among multipartite entangled states, Dicke states and more general symmetric states form a natural family characterized by permutation symmetry. The $n$-qubit Dicke state with $k$ excitations is the uniform superposition of all $n$-bit strings of Hamming weight $k$, namely
\begin{equation}\label{eq:Dicke}
    \ket{D_k^n}=\binom{n}{k}^{-1/2}\sum_{\substack{x\in\{0,1\}^n\\ |x|=k}}\ket{x},
\end{equation}
where $|x|=\sum_{i=1}^{n}x_i$ denotes the Hamming weight of the bit string $x$. These states exhibit genuine multipartite entanglement~\cite{Zhou2019decomposition} and have been widely studied as resources for quantum metrology, where they can enable sensitivities beyond the standard quantum limit~\cite{Pezze2018metrology,Holland1993Heisenberg}, as well as for quantum networks~\cite{Prevedel2009network,Chiuri2012network} and entanglement distribution~\cite{Roga2023distribution}. They have also acquired a direct algorithmic role: decoded quantum interferometry begins by preparing a prescribed superposition of Dicke states and uses it as the interference resource for structured optimization~\cite{Jordan2025DQI}. More generally, symmetric states are precisely the $n$-qubit states invariant under arbitrary permutations of the qubits. Since the Dicke states form an orthonormal basis of the symmetric subspace, every symmetric pure state can be written as 
\begin{equation}\label{eq:symmetric}
    \ket{\psi}=\sum_{k=0}^{n}a_k\ket{D_k^n}.
\end{equation}
Preparing such states is relevant for tasks and platforms governed by collective degrees of freedom~\cite{marconi2025symmReview}, including collective-spin probes in quantum metrology~\cite{Pezze2018metrology}, permutation-invariant encodings~\cite{Ouyang2014permutationInv,Ouyang2016permutationInv}, and photonic or spin-ensemble implementations of multipartite entanglement~\cite{Kiesel2010multiEntangle}.

Despite their broad relevance, preparing Dicke and symmetric states at low quantum depth remains challenging. In the standard unitary circuit model with bounded-fan-in gates, establishing correlations across the entire system generally requires nonconstant depth, and logarithmic-depth lower bounds arise for representative long-range-entangled states~\cite{Bartschi_2022_dicke,yuan2025depthefficientquantumcircuitsynthesis,Cruz_2019_ghzw}. Intermediate measurements and classical feedforward can circumvent this light-cone restriction by inducing and subsequently correcting long-range correlations~\cite{liu2025lowdepthquantumsymmetrization,zi2025constant,liu2025state}. Equivalently, adaptive circuits can implement nonlocal primitives such as unbounded quantum fan-out in constant quantum depth~\cite{takahashi2016collapse}, enabling a broad class of shallow state-preparation protocols~\cite{Piroli_2024_approx,buhrman2024state,zi2025constant}. Several low-depth approaches to Dicke-state preparation based on adaptive circuits or unbounded quantum fan-out gates have been developed, as summarized in Table~\ref{tab:dicke}. 

In this work, we develop a constant-depth adaptive approach based on the uniform subset superposition (USS),
\begin{equation}
    \ket{\mathrm{USS}_{n,k}}=\binom{n}{k}^{-1/2}\sum_{0\leq j_1<\cdots<j_k\leq n-1}
    \ket{j_1}\cdots\ket{j_k},
\end{equation}
which coherently encodes all $k$-element subsets by their increasingly ordered elements. Besides its role in our construction, the USS is the uniform vertex state of the Johnson graph underlying quantum-walk algorithms for element distinctness and subset finding~\cite{ambainis2007quantumwalk,childs2005quantumalgorithms}. Beyond these cases, the same uniform superposition over fixed-size subsets serves as the uniform set-register initial state in Ambainis-style algorithms for graph collision, triangle finding, and related graph-property problems~\cite{magniez2007triangle,Jeffery2013walk}. It also coincides with the compactly encoded uniform superposition over candidate simplices used in quantum algorithms for topological data analysis~\cite{McArdle2026streamlined}. We show that the USS can be prepared exactly in constant adaptive depth using $O\left(k^2\log^2 n\right)$ ancillary qubits and polynomial-time classical processing, with success probability at least $1/k$. This yields an exact constant-depth Dicke-state preparation protocol using $O\left(n^2+k^2\log^2 n+kn\log n\log\log n\right)$ ancillary qubits, with parallel repetition reducing the failure probability exponentially. Our construction provides an adaptive ordering-based alternative to deterministic Grover-search approaches. Comparisons with previous works are listed in Table~\ref{tab:dicke}.

We further establish a general lifting framework for arbitrary symmetric-state preparation. Rather than designing a separate circuit for each symmetric state, our framework shows how to coherently combine a family of clean unitary Dicke-state preparation circuits. If every Dicke state can be prepared with depth $O(L)$ using at most $O(\xi)$ ancillary qubits, then an arbitrary $n$-qubit symmetric state can be prepared with also $O(L)$ depth and with $O(n^2\log n+n\xi)$ ancillary qubits. This separates the task of coherently combining different Hamming-weight sectors from that of preparing each individual Dicke state, and implies that any improvement in clean unitary Dicke-state preparation immediately leads to a corresponding improvement for general symmetric states. Combining this framework with recent constant-depth unitary Dicke-state constructions~\cite{joshi2026constantdepthunitarypreparationdicke}, we obtain an exact constant-depth preparation scheme for arbitrary $n$-qubit symmetric states using $O(n^3\sqrt{\log n})$ ancillary qubits. Comparisons with previous works are summarized in Table~\ref{tab:symm}.

The rest of the paper is organized as follows. In Section~\ref{sec:pre}, we introduce the adaptive circuit model and the quantum primitives used in our constructions. In Section~\ref{sec:dicke}, we present the constant-depth preparation of the USS state and its conversion into arbitrary Dicke states. In Section~\ref{sec:sym}, we establish the lifting framework from unitary Dicke-state preparation to arbitrary symmetric-state preparation. Section~\ref{sec:con_dis} concludes with a discussion of applications and open problems. Appendix~\ref{app:prob} provides a detailed analysis of the descent-number distribution and the associated success probability.

\begin{table}[t]
\centering
\caption{Comparison of Dicke-state preparation circuits.}
\label{tab:dicke}

\renewcommand{\arraystretch}{1.25}
\setlength{\tabcolsep}{6pt}

\resizebox{\linewidth}{!}{
\begin{tabular}{c|c|c|c|c|c|c|c}
\hline
\multirow{2}{*}{\textbf{Paper}}
& \multicolumn{5}{c|}{\textbf{Circuit}}
& \multicolumn{2}{c}{\textbf{Dicke State Output}} \\
\cline{2-8}
& \textbf{Type}
& \textbf{Connectivity}
& \textbf{Interactions}
& \textbf{Depth}
& \textbf{Ancillae}
& \textbf{Type}
& \textbf{Weight $(k)$} \\
\hline

\cite{B_rtschi_2019}
& Unitary
& 1D Chain
& $O(1)$-Width
& $O(n)$
& $0$
& E
& $[1,n/2]$ \\

\cite{Bartschi_2022_dicke}
& Unitary
& $(n_1\times n_2)$-Grid
& $O(1)$-Width
& $O(\sqrt{nk})$
& $0$
& E
& $[n_2/n_1,n/2]$ \\

\cite{yuan2025depthefficientquantumcircuitsynthesis}
& Unitary
& $(n_1\times n_2)$-Grid
& $O(1)$-Width
& $O\!\left(k\log\frac{n}{k}+n_2\right)$
& $0$
& E
& $[n_2/n_1,n/2]$ \\

\cite{yuan2025depthefficientquantumcircuitsynthesis}
& Unitary
& $(n_1\times n_2)$-Grid
& $O(1)$-Width
& $O(n_2)$
& $0$
& E
& $[1,n_2/n_1]$ \\
\hline
\noalign{\vskip 3pt}
\hline

\cite{Cruz_2019_ghzw}
& Unitary
& All-to-All
& $O(1)$-Width
& $O(\log n)$
& $0$
& E
& $1$ \\

\cite{Bartschi_2022_dicke}
& Unitary
& All-to-All
& $O(1)$-Width
& $O\!\left(k\log\frac{n}{k}\right)$
& $0$
& E
& $[1,n/2]$ \\

\cite{liu2025lowdepthquantumsymmetrization}
& Unitary
& All-to-All
& $O(1)$-Width
& $O\!\left(\log^3 n\cdot\log\log n\right)$
& $O(n\log n)$
& E
& $[1,n/2]$ \\

\cite{liu2025lowdepthquantumsymmetrization}
& Prob.
& All-to-All
& $O(1)$-Width
& $O\!\left(\log n\cdot\log\log n\right)$
& $O(n\log n)$
& E
& $[1,n/2]$ \\

\cite{yuan2025depthefficientquantumcircuitsynthesis}
& Unitary
& All-to-All
& $O(1)$-Width
& $O\!\left(\log k\log\frac{n}{k}+k\right)$
& $0$
& E
& $[1,n/2]$ \\
\hline
\noalign{\vskip 3pt}
\hline

\cite{buhrman2024state}
& LAQCC
& Local Grid
& $O(1)$-Width
& $O(1)^{*}$
& $O(n\log n)$
& E
& $1$ \\

\cite{buhrman2024state}
& LAQCC
& Local Grid
& $O(1)$-Width
& $O(1)^{*}$
& $O(n^2\log n)$
& E
& $O(\sqrt{n})$ \\

\cite{buhrman2024state}
& LAQCC
& 1D Chain
& $O(1)$-Width
& $O(\log n)^{*}$
& $\Omega(n^2)$
& E
& $[1,n/2]$ \\

\cite{Yu_2026_dicke}
& Adaptive
& All-to-All
& $O(1)$-Width
& $\operatorname{polylog}(n)^{*}$
& $O(\log n)$
& E
& $[1,n/2]$ \\
\hline
\noalign{\vskip 3pt}
\hline

\cite{farrell2025digital}
& LOCC
& Grid
& $O(1)$-Width
& $O(1)^{*}$
& $0$
& A
& $1$ \\

\cite{Piroli_2024_approx}
& LOCC
& Grid
& $O(1)$-Width
& $O(\log k)^{*}$
& $O(1)$
& A
& $[1,n/2]$ \\

\cite{Piroli_2024_approx}
& LOCC
& Grid
& $O(1)$-Width
& $O(1)^{*}$
& $O(\log k)$
& A
& $[1,n/2]$ \\
\hline
\noalign{\vskip 3pt}
\hline

\cite{grier2026qac0}
& Unitary
& All-to-All
& Global CZ
& $O(1)$
& $\operatorname{poly}(n)$
& E
& $1$ \\

\cite{joshi2026constantdepthunitarypreparationdicke}
& Unitary
& All-to-All
& Global CZ
& $O(1)$
& $O(n^{k+1})$
& E
& $O(1)$ \\

\cite{joshi2026constantdepthunitarypreparationdicke}
& Unitary
& All-to-All
& Global CZ
& $O(1)$
& $O(1)$
& A
& $1$ \\

\cite{joshi2026constantdepthunitarypreparationdicke}
& Unitary
& All-to-All
& Global FAN-OUT
& $O(1)$
& $O\!\left(n^2\sqrt{\log n}\right)$
& E
& $[1,n/2]$ \\
\hline
\noalign{\vskip 3pt}
\hline

Thm.~\ref{thm:Dicke}
& Adaptive
& All-to-All
& $O(1)$-Width
& $O(1)^{*}$
& $O\!\left(
    n^{2}
    + k^{2}\log^{2}n
    + kn\log n\log\log n
  \right)$
& E
& $[1,n/2]$ \\

Cor.~\ref{cor:dicke}
& Adaptive
& All-to-All
& $O(1)$-Width
& $O(1)^{*}$
& $O\!\left(
    n^{2}
    + c\alpha k^{2}\log^{2}n
    + kn\log n\log\log n
  \right)$
& E
& $[1,n/2]$ \\
\hline
\end{tabular}
}

\vspace{4pt}
\begin{minipage}{\textwidth}
\footnotesize
\textit{Notes.}
The table is organized primarily following the summary in
Ref.~\cite{joshi2026constantdepthunitarypreparationdicke}, supplemented by
our results. The circuit types include Unitary, LAQCC (Local Adaptive Quantum
Circuits with Classical Communication~\cite{buhrman2024state}), LOCC (Local Operations and Classical Communication), Prob. (Probabilistic), and
Adaptive. The labels E and A denote exact and bounded-error approximate
preparation, respectively. For non-unitary circuit models involving nontrivial classical post-processing (having a ($\ast$) following the depth),
the reported quantum depth excludes the additional classical circuit depth.
Our protocols are probabilistic. Theorem~\ref{thm:Dicke} succeeds with
probability at least $1/k$, whereas Corollary~\ref{cor:dicke} has a failure
probability exponentially small in the constant $c$. Here,
$\alpha\leq k$ and is expected to scale as $O(\sqrt{k})$ according to
Appendix~\ref{app:prob}. Conditioned on success, both protocols are exact.
Although our protocols use quantum fan-out gates, these gates can be synthesized
from bounded-fan-in gates and mid-circuit measurements in constant depth, as
shown in Section~\ref{sec:pre}. Because fan-out is not treated as a fundamental
primitive, the interaction width is reported as $O(1)$.
\end{minipage}

\end{table}

\begin{table}[t]
\centering
\caption{Comparison of symmetric-state preparation circuits.}
\label{tab:symm}

\renewcommand{\arraystretch}{1.25}
\setlength{\tabcolsep}{6pt}

\resizebox{\linewidth}{!}{
\begin{tabular}{c|c|c|c|c|c|c|c}
\hline
\multirow{2}{*}{\textbf{Paper}}
& \multicolumn{5}{c|}{\textbf{Circuit}}
& \multicolumn{2}{c}{\textbf{Symmetric State Output}} \\
\cline{2-8}
& \textbf{Type}
& \textbf{Connectivity}
& \textbf{Interactions}
& \textbf{Depth}
& \textbf{Ancillae}
& \textbf{Type}
& \textbf{Support} \\
\hline

\cite{B_rtschi_2019}
& Unitary
& 1D Chain
& $O(1)$-Width
& $O(n)$
& $0$
& E
& Arbitrary \\

\cite{Bartschi_2022_dicke}
& Unitary
& All-to-All
& $O(1)$-Width
& $O\!\left(k\log\frac{n}{k}\right)$
& $0$
& E
& $[0,k]$ \\

\cite{Bartschi_2022_dicke}
& Unitary
& Grid
& $O(1)$-Width
& $O(\sqrt{nk})$
& $0$
& E
& $[0,k]$ \\
\hline
\noalign{\vskip 3pt}
\hline

\cite{liu2025lowdepthquantumsymmetrization}
& Unitary
& All-to-All
& $O(1)$-Width
& $O(\log^{3}n\cdot\log\log n)$
& $O(n\log n)$
& E
& Arbitrary\footnotemark[1] \\

\cite{bond2025global}
& Variational
& Collective
& Global OAT + rotations
& $\left\lceil 2n/3\right\rceil+O(1)$
& $0$
& Var.
& Arbitrary \\
\hline
\noalign{\vskip 3pt}
\hline

\cite{gretta2026qac0}
& Unitary
& All-to-All
& Global Toffoli
  $+\operatorname{FANOUT}_{k}$
& $O(1)$
& $\operatorname{poly}(n)$
& E
& $[0,k]$ \\

\cite{gretta2026qac0}
& Unitary
& All-to-All
& Global Toffoli
  $+\operatorname{FANOUT}_{n}$
& $O(1)$
& $\operatorname{poly}(n)$
& E
& Arbitrary \\
\hline
\noalign{\vskip 3pt}
\hline

Thm.~\ref{thm:symm}
& Unitary
& All-to-All
& Global FAN-OUT
& $O(L)$
& $O\!\left(n^{2}\log n+n\xi\right)$
& E
& Arbitrary \\

Cor.~\ref{cor:symmetric}
& Unitary
& All-to-All
& Global FAN-OUT
& $O(1)$
& $O\!\left(n^{3}\sqrt{\log n}\right)$
& E
& Arbitrary \\
\hline

\end{tabular}
}
\footnotetext[1]{The quoted complexity assumes an input state of the form
$\sum_{k=0}^{n}a_k|0^{n-k}1^k\rangle$;
the cost of preparing the coefficients $\{a_k\}$ is not included.}

\vspace{4pt}
\begin{minipage}{\textwidth}
\footnotesize
\textit{Notes.}
The labels E and Var.\ denote exact and variational preparations,
respectively.
The support $[0,k]$ means that the output may be an arbitrary state
$\sum_{j=0}^{k}a_j\lvert D_j^n\rangle$ supported on Dicke weights at most $k$. The quantities $L$ and $\xi$ in Theorem~\ref{thm:symm} denote, respectively, the
depth and ancillary-qubit complexity of the underlying clean unitary Dicke-state preparation circuits.
\end{minipage}

\end{table}

\section{Preliminary}\label{sec:pre}

\subsection{Notations and Adaptive Circuits}
In this subsection, we introduce some basic notations and the circuit model used throughout this work. The Hilbert space of an $n$-qubit system is defined as $\mathcal{H}=\bigotimes_{i=1}^{n}\mathcal{H}_i$, where each $\mathcal{H}_i$ is a two-dimensional Hilbert space spanned by computational basis $\{\ket{0},\ket{1}\}$. A quantum gate acting on $\mathcal{H}$ is a unitary operation $U$. Our elementary quantum operations consist of bounded-fan-in unitary gates and single-qubit measurements. For convenience, we will also use the unbounded quantum fan-out gate as a derived primitive; as discussed below, it admits a constant-depth implementation in the adaptive circuit model.

We use the standard conventions for single-qubit Pauli matrices:
\begin{equation}
    X=\begin{pmatrix}
        0 & 1 \\ 1 & 0
    \end{pmatrix},\qquad
    Y=\begin{pmatrix}
        0 & -i \\ i & 0
    \end{pmatrix},\qquad
    Z=\begin{pmatrix}
        1 & 0 \\ 0 & -1
    \end{pmatrix}.
\end{equation}
We also use the Hadamard gate $H$, the phase gate $S$, and the $\pi/8$ gate $T$, defined by
\begin{equation}
    H=\frac{1}{\sqrt{2}}\begin{pmatrix}
        1 & 1 \\ 1 & -1
    \end{pmatrix},\qquad
    S=\begin{pmatrix}
        1 & 0 \\ 0 & i
    \end{pmatrix},\qquad
    T=\begin{pmatrix}
        1 & 0 \\ 0 & e^{i\pi/4}
    \end{pmatrix}.
\end{equation}
The two-qubit $\CNOT$ gate acts on computational-basis states as
\begin{equation}
    \CNOT:\ket{x}\ket{y}\rightarrow\ket{x}\ket{y\oplus x}.
\end{equation}
$\{\CNOT,H,T\}$ together forms a universal gate set, and $\{\CNOT,H,S\}$ generates the Clifford group.

Unless otherwise specified, all data and ancilla qubits are initialized in the state $\ket{0}$ and measured in the computational basis. An adaptive quantum circuit consists of alternating layers of bounded-fan-in unitary gates and intermediate measurements. Measurement outcomes may be processed classically and used to determine subsequent operations through classical feedforward, after which the measured qubits may be discarded. We define the adaptive quantum depth as the total number of unitary and measurement layers along the longest execution path. Classical processing between quantum layers is not included in this quantum-depth measure, and its computational complexity will be analyzed separately. The principal resources considered in this work are therefore the adaptive quantum depth and the total number of ancilla qubits.

Adaptive circuits are particularly useful for preparing states with long-range correlations~\cite{Cruz_2019_ghzw,zi2025constant}. A typical example is the GHZ states $\ket{\mathrm{GHZ}_n}\equiv\frac{1}{\sqrt{2}}\left(\ket{0}^{\otimes n}+\ket{1}^{\otimes n}\right)$. Such a state can be directly prepared by applying the unbounded quantum fan-out gate
\begin{equation}\label{eq:fanout}
    F_{n-1}:\ket{x_0}\ket{x_1,\ldots,x_{n-1}}\rightarrow\ket{x_0}\ket{x_1\oplus x_0,\ldots,x_{n-1}\oplus x_0},
\end{equation}
on the $\ket{+}\ket{0}^{\otimes (n-1)}$ state, where $\ket{x_i}$'s for $0\leq i\leq {n-1}$ are computational basis states.

Any unitary implementation of $F_m$ using only bounded-fan-in gates requires depth $\Omega(\log m)$~\cite{watts2019exponential}. By contrast, when intermediate measurements and classical feedforward are allowed, the fan-out operation can be implemented in constant adaptive depth~\cite{Cruz_2019_ghzw}. The underlying idea is to first prepare constant-size GHZ blocks in parallel and then connect these blocks using a single round of intermediate measurements. The measurement outcomes determine a final layer of Pauli corrections that aligns the relative branches of the local GHZ states, thereby producing a globally correlated state. An explicit example of this construction for a nine-qubit GHZ state is shown in Fig.~\ref{fig:9GHZ}.

\begin{figure}[t]
\centering
\captionsetup{justification=justified, singlelinecheck=false}
\begin{quantikz}
  \lstick{$\ket{\psi}$} && \ctrl{2}\slice{1} &&& \slice{2} && \\
  \lstick{$\ket{0}$} && \targ{} &&&&& \\
  \lstick{$\ket{0}$} && \targ{} & \ctrl{1} &&&& \\
  \lstick{$\ket{0}$} &&& \targ{} & \targ{} & \meter{c_1} & \ctrl[vertical wire=c]{3}\setwiretype{c}\wire[d]{c} \\
  \lstick{$\ket{0}$} & \gate{H} & \ctrl{2} && \ctrl{-1} && \gate{X^{c_1}}\wire[d]{c} & \\
  \lstick{$\ket{0}$} && \targ{} &&&& \gate{X^{c_1}}\wire[d]{c} & \\
  \lstick{$\ket{0}$} && \targ{} & \ctrl{1} &&& \gate{X^{c_1}}\wire[d]{c} & \\
  \lstick{$\ket{0}$} &&& \targ{} & \targ{} & \meter{c_2} & \ctrl[vertical wire=c]{3}\setwiretype{c}\wire[d]{c} \\
  \lstick{$\ket{0}$} & \gate{H} & \ctrl{2} && \ctrl{-1} && \gate{X^{c_1\oplus c_2}}\wire[d]{c} & \\
  \lstick{$\ket{0}$} && \targ{} &&&& \gate{X^{c_1\oplus c_2}}\wire[d]{c} & \\
  \lstick{$\ket{0}$} && \targ{} &&&& \gate{X^{c_1\oplus c_2}} &
\end{quantikz}
\caption{The quantum circuit for preparing $9$-qubit GHZ states given in~\cite{zi2025constant}. This circuit first locally prepares constant-weight GHZ states before slice $1$. Then, before slice $2$, the $\CNOT$ gates and the measurements $c_1$, $c_2$ are used to establish long-range correlations between neighboring local blocks. Finally, one layer of post-measurement operations is applied to correct some flips.}
\label{fig:9GHZ}
\end{figure}

In this work, we mainly consider adaptive circuit preparation tasks for Dicke states, defined in Eq.~\eqref{eq:Dicke} and symmetric states. By symmetric states, we refer to the states that remain unchanged when the order of qubits is permuted. Mathematically, let $\Pi_\pi$ denote the unitary operator that permutes the qubits according to $\pi\in S_n$. An $n$-qubit pure state $\ket{\psi}$ is permutation symmetric if $\Pi_{\pi}\ket{\psi}=\ket{\psi}$ for any permutation operator $\Pi_{\pi}$. Since the Dicke states $\{\ket{D_k^n}\}_{k=0}^{n}$ form an orthonormal basis of the symmetric subspace, every symmetric pure state $\ket{\psi}$ admits the expansion in Eq.~\eqref{eq:symmetric}.

\subsection{Useful Quantum Gates Based on Unbounded Quantum Fan-out}
The unbounded quantum fan-out gate defined in Eq.~\eqref{eq:fanout} is a powerful primitive: when combined with bounded-fan-in gates, it enables constant-depth implementations of a broad class of nonlocal quantum operations. Since unbounded fan-out itself admits a constant-depth implementation in the adaptive circuit model considered here, each construction summarized below can be converted into a constant-depth adaptive circuit consisting only of bounded-fan-in gates, intermediate measurements, and classical feedforward. Table~\ref{tab:gates} lists the gates used in our protocols, with their actions specified on computational-basis states.

\begin{table}[t]
\centering
\caption{Useful quantum gates based on unbounded quantum fan-out.}
\label{tab:gates}

\renewcommand{\arraystretch}{1.25}
\setlength{\tabcolsep}{6pt}

\resizebox{\linewidth}{!}{
\begin{tabular}{c|l|c|c|c}
\hline

\textbf{Gate}
& \multicolumn{1}{c|}{\textbf{Operation}}
& \textbf{Width}
& \textbf{Type}
& \textbf{Implementation} \\
\hline
    
$\text{OR}_n$ 
& $\ket{x_1}\cdots\ket{x_n}\ket{x_0}\rightarrow \ket{x_1}\cdots\ket{x_n}\ket{x_0\oplus \text{OR}_n(x)}$ 
& $O(n\log n)$ 
& E
& \cite{takahashi2016collapse}, Theorem 1.1 \\

$\text{AND}_n$ 
& $\ket{x_1}\cdots\ket{x_n}\ket{x_0}\rightarrow \ket{x_1}\cdots\ket{x_n}\ket{x_0\oplus \text{AND}_n(x)}$ 
& $O(n\log n)$ 
& E
& directly from OR$_n$ \\

$\text{Equal}_i$ 
& $\ket{j}\ket{x_0}\rightarrow\begin{cases}
    \ket{j}\ket{x_0\oplus 1}, & \text{if }j=i \\
    \ket{j}\ket{x_0}, & \text{else}
\end{cases}$ 
& $O(n\log n)$ 
& E
& directly from OR$_n$ \\

$\text{Exact}_t$ 
& $\ket{x_1}\cdots\ket{x_n}\ket{x_0}\rightarrow \ket{x_1}\cdots\ket{x_n}\ket{x_0\oplus\mathbb{I}_{|x|=t}}$ 
& $O(n\log n)$ 
& E
& directly from OR$_n$ \\

Hammingweight 
& $\ket{x_1}\cdots\ket{x_n}\ket{0}\rightarrow \ket{x_1}\cdots\ket{x_n}\ket{|x|}$ 
& $O(n^2)$ 
& E
& \cite{takahashi2016collapse}, Lemma 4.2 \\
    
&
& $O(n\log n)$ 
& A 
& \cite{hoyer2005quantum}, Theorem 6.9 \\

$\text{Threshold}_t$ 
& $\ket{x_1}\cdots\ket{x_n}\ket{x_0}\rightarrow \ket{x_1}\cdots\ket{x_n}\ket{x_0\oplus\mathbb{I}_{|x|\geq t}}$ 
& depends on $t$ 
& E
& \cite{takahashi2016collapse}, Theorem 1.2 \\

&
& $O(n\log n)$ 
& A 
& \cite{hoyer2005quantum}, Remark 6.10 \\

$\text{Add}_n$ 
& $\ket{x}\ket{y}\rightarrow\ket{x}\ket{y+x\text{ mod }2^n}$ 
& $O(n^2)$ 
& E
& AC$^0$ circuit \\

Equality 
& $\ket{x}\ket{y}\ket{0}\rightarrow\ket{x}\ket{y}\ket{\mathbb{I}_{x=y}}$ 
& $O(n^2)$ 
& E
& \cite{buhrman2024state}, Appendix B.2 \\

Greaterthan 
& $\ket{x}\ket{y}\ket{0}\rightarrow\ket{x}\ket{y}\ket{\mathbb{I}_{x>y}}$ 
& $O(n^2)$ 
& E
& \cite{buhrman2024state}, Appendix B.3 \\
\hline
\end{tabular}
}

\vspace{4pt}
\begin{minipage}{\textwidth}
\footnotesize
\textit{Notes.}
Those gates are all realized by constant-depth adaptive circuits, but with different number of ancillae. The width denotes the total number of qubits involved, including the input, output, and ancillary registers. The labels E and A indicate exact and bounded-error approximate implementations, respectively. For the exact implementation of $\mathrm{Threshold}_t$, the width is $O(n\log n)$ for any $1\le\tilde{t}\le\log n$ and $O\left(n\sqrt{\tilde{t}\log n}\right)$ for any $\log n\le\tilde{t}\le\frac{n}{2}$, where $\tilde{t}=\min\{t,n-t\}$.
\end{minipage}

\end{table}

In Table~\ref{tab:gates}, $x$, $y$ and $j$ are short for the bit strings $(x_1\cdots x_n)$, $(y_1\cdots y_n)$ and $(j_1\cdots j_n)$, and $\mathbb{I}_A$ is an indicator which takes value $1$ if the Boolean expression $A$ is true, and $0$ otherwise. The construction for the $\text{OR}_n$ gate is provided in~\cite{takahashi2016collapse}, upon which the circuits for the $\text{AND}_n$ gate and $\text{Equal}_i$ gate are naturally built. The $\text{Exact}_t$ can be implemented through a slight modification in the circuit for $\text{OR}_n$, as noted at the end of Sec.~3 of Ref.~\cite{takahashi2016collapse}. The same reference provides exact constant-depth constructions for the counting operation, which implements $\mathrm{Hammingweight}$, and for $\mathrm{Threshold}_t$. However, if we allow an inverse-polynomial approximation error, the width for implementing $\text{Hammingweight}$ and $\text{Threshold}_t$ can be reduced, as shown in~\cite{hoyer2005quantum}. Equipped with $\text{OR}_n$ and $\text{AND}_n$, the $\text{Add}_n$ can be constructed in constant depth simulating the classical AC$^0$ architecture, based on which the Equality gate and the Greaterthan gate arise, as shown in Appendix~B of~\cite{buhrman2024state}.
Throughout this work, the Greaterthan, Equality, $\text{Equal}_i$ and Hammingweight gates will be used frequently.

We note a discrepancy in the resource accounting of Ref.~\cite{buhrman2024state}. Table~5 of that reference assigns width $O(n\log n)$ to the exact $\mathrm{Hammingweight}$ gate and cites Lemma~4 of Ref.~\cite{takahashi2016collapse}. The cited lemma, however, establishes an exact constant-depth counting circuit of size $O(n^2)$, and the fully parallel construction given in its proof uses $O(n^2)$ qubits. Thus, the quoted $O(n\log n)$ width does not follow from the cited construction. We therefore use the exact width $O(n^2)$ throughout this work.
This observation does not preclude the existence of a different exact construction with width $O(n\log n)$, but such a construction is not provided by the cited result.

\section{A Constant-Depth Preparation Scheme for Dicke States}\label{sec:dicke}
In this section, we present a constant-depth adaptive protocol for preparing Dicke states. At a schematic level, the protocol can be written as
\begin{equation}
    \text{Circuit}_{Dicke}=Q_2\circ C_2\circ M\circ Q_1\circ C_1,
    \label{eq:circuit}
\end{equation}
where the operations are applied from right to left. Here, $C_1$ and $C_2$ denote classical computations, $Q_1$ and $Q_2$ denote constant-depth adaptive quantum subcircuits, and $M$ is the distinguished intermediate measurement layer used in the preparation of the index register introduced below. Both classical stages can be performed in polynomial time, while the total adaptive quantum depth of $Q_1$, $M$, and $Q_2$ is $O(1)$. The layouts of $Q_1$ and $Q_2$ may depend on the outputs of $C_1$ and $C_2$, respectively. The subcircuits $Q_1$ and $Q_2$ may themselves contain the intermediate measurements required to implement quantum fan-out gates; these measurements are deterministic components of the corresponding adaptive implementations and are distinct from the measurement layer $M$ responsible for the heralded part of our protocol.

We first state the main result of this section.

\begin{theorem}[Constant-Depth Dicke-State Preparation]\label{thm:Dicke}
    For any $n$ and $0\leq k\leq n$, the Dicke-$(n,k)$ state can be exactly prepared using $O(1)$ depth quantum computation with $O\left(n^2+k^2\log^2 n+kn\log n\log\log n\right)$ ancillary qubits and polynomial-time classical computation, with success probability at least $1/k$.
\end{theorem}

Ref.~\cite{buhrman2024state} gave a constant-depth adaptive construction for $k=O(\sqrt{n})$ using $O(n^3)$ ancillary qubits, together with an $O(\log n)$-depth construction for arbitrary $k$. In our protocol, for arbitrary $0\le k\le n$, the ancillary-qubit complexity in Theorem~\ref{thm:Dicke} remains below $O(n^2\log^2 n)$, which is an improvement to the constructions in Ref.~\cite{buhrman2024state}. More recently, Ref.~\cite{joshi2026constantdepthunitarypreparationdicke} obtained a constant-depth unitary construction for arbitrary excitation number using $O(n^2\sqrt{\log n})$ ancillary qubits. In our protocol, when $k$ is of the order $O(n/\log^{3/4}n)$, the total resource estimation is below $O(n^2\sqrt{\log n})$. The purpose of Theorem~\ref{thm:Dicke} is therefore to provide a constant-depth protocol valid for all $k$ (although not the first one) with SOTA resource estimation at a specific range, and to establish an alternative construction based on a substantially different mechanism.

In particular, the construction of Ref.~\cite{joshi2026constantdepthunitarypreparationdicke} achieves arbitrary-weight Dicke-state preparation through a unitary procedure involving deterministic Grover search. By contrast, our protocol replaces the Grover-search-based filtering step by an adaptive ordering procedure: intermediate measurements and classical post-processing are used to transform an initially unordered collection of indices into a coherent superposition of ordered, distinct indices. This leads naturally to the uniform subset superposition introduced below. The two approaches thus provide complementary routes to constant-depth Dicke-state preparation, with different circuit structures and resource trade-offs, while both remain within the regime of polynomial ancillary resources.

To describe the protocol and prove Theorem~\ref{thm:Dicke}, we introduce a family of states that forms the central ingredient of our construction. These states are coherent superpositions over all $k$-element subsets of $\{0,\ldots,n-1\}$, represented by their increasingly ordered elements, and may be of independent interest beyond Dicke-state preparation.

\begin{definition}[The Uniform Subset Superposition]\label{def:uss}
For any integer $0\le k\le n$, the uniform subset superposition (USS) is defined as:
\begin{equation}
    \ket{\mathrm{USS}_{n,k}}=\frac{1}{\sqrt{\binom nk}}\sum_{0\leq j_1<\cdots<j_k\leq n-1}\ket{j_1}\cdots\ket{j_k}.
\end{equation}
Each index $j_\ell$ is stored in a
$\lceil\log_2 n\rceil$-qubit computational-basis register. For $k=0$, the USS is understood as the state of an empty index register.
\end{definition}

Preparing $\ket{\mathrm{USS}_{n,k}}$ is the main technical component of our Dicke-state preparation protocol. Once this ordered index state is available, it can be converted into a Dicke state using only additional constant-depth quantum operations.

Let $I$ denote the index register, let $T$ denote an $n$-qubit target register, and set $L=\lceil\log_2 n\rceil$. Then the structure of the protocol is
\begin{equation}\label{eq:procedure}
\begin{aligned}
    \ket{0}_I^{\otimes kL}\ket{0}_T^{\otimes n}
    &\xrightarrow{\hspace{1em}1\hspace{1em}}
    \ket{\mathrm{USS}_{n,k}}_I\ket{0}_T^{\otimes n}
    \\
    &\xrightarrow{\hspace{1em}2\hspace{1em}}
    \frac{1}{\sqrt{\binom{n}{k}}}
    \sum_{0\leq j_1<\cdots<j_k\leq n-1}
    \ket{j_1,\ldots,j_k}_I
    \ket{\bigoplus_{\ell=1}^{k}e_{j_\ell}}_T
    \\
    &\xrightarrow{\hspace{1em}3\hspace{1em}}
    \ket{0}_I^{\otimes kL}
    \frac{1}{\sqrt{\binom{n}{k}}}
    \sum_{0\leq j_1<\cdots<j_k\leq n-1}
    \ket{\bigoplus_{\ell=1}^{k}e_{j_\ell}}_T=
    \ket{0}_I^{\otimes kL}\ket{D_k^n}_T,
\end{aligned}\end{equation}
where $e_j\in\{0,1\}^n$ denotes the one-hot string whose only nonzero entry is at position $j$. Since the indices $j_1,\ldots,j_k$ are distinct, the bitwise sum
$\bigoplus_{\ell=1}^{k}e_{j_\ell}$ has Hamming weight exactly $k$.

In step~1 of Eq.~\eqref{eq:procedure}, we prepare the uniform subset superposition and use it as an index register encoding the locations of the excitations. In step~2, these indices are used to write $1$ into the corresponding positions of the target register. In step~3, the index register is coherently erased, leaving the desired Dicke state in the target register. The strict ordering $j_1<\cdots<j_k$ is crucial: it ensures that every weight-$k$ target string corresponds to a unique index tuple.

\medskip
\noindent\textbf{Preparing the Uniform Subset Superposition.}
To prove Theorem~\ref{thm:Dicke}, we first recall a result from Ref.~\cite{buhrman2024state} concerning the preparation of a uniform superposition over an arbitrary number of computational-basis states.

\begin{lemma}\label{lem:superposition}
    \textnormal{\cite{buhrman2024state}}
    For any positive integer $q$, there exists a deterministic constant-depth adaptive quantum circuit that prepares $\frac{1}{\sqrt{q}}\sum_{i=0}^{q-1}\ket{i}$ using $O(\lceil\log_2 q\rceil^2)$ qubits.
\end{lemma}

The construction may be summarized as follows. Let $\ell=\lceil\log_2 q\rceil$. One first applies Hadamard gates to prepare the uniform superposition $\frac{1}{\sqrt{2^\ell}}\sum_{i=0}^{2^\ell-1}\ket{i}$, and then uses exact amplitude amplification, implemented through deterministic Grover search~\cite{long2001grover}, to restrict the support to the first $q$ basis states. Equipped with Lemma~\ref{lem:superposition} and the constant-depth primitives summarized in Table~\ref{tab:gates}, we now establish one of the main technical results of this work.

\begin{lemma}[Constant-depth preparation of the USS]
\label{lem:index}
    For any integer $0\leq k\leq n$, the state $\ket{\mathrm{USS}_{n,k}}$ can be prepared exactly by a constant-depth adaptive quantum circuit using $O\left(k^2\log^2 n\right)$ ancillary qubits and polynomial-time classical processing, with success probability at least $1/k$.
\end{lemma}

\begin{proof}
The cases $k=0$ and $k=1$ are immediate, so we assume throughout the proof that $2\leq k\leq n$. All states below are written without their normalization factors.

Let $\eta$ be an integer satisfying $n-k+1\leq\eta\leq n$, whose value will be chosen later. Using Lemma~\ref{lem:superposition}, we prepare $k$ independent uniform index registers, $\sum_{j_1,\ldots,j_k=0}^{\eta-1}\ket{j_1}\cdots\ket{j_k}$. Each index is represented by $O(\log n)$ qubits. Preparing the $k$ registers in parallel therefore requires constant adaptive depth and $O\left(k\log^2 n\right)$ ancillae according to Lemma~\ref{lem:superposition}.

Next, we do the following ordering:
\begin{equation}\begin{aligned}
    \sum_{j_1,\cdots,j_k=0}^{\eta-1}\ket{j_1}\cdots\ket{j_k}
    &\xrightarrow{\hspace{1em}1\hspace{1em}}\sum_{j_1,\cdots,j_k=0}^{\eta-1}\begin{aligned}[t]\label{eq:get_order}
        \ket{j_1}&\ket{\mathbb{I}_{j_1>j_2}}\ket{\mathbb{I}_{j_1>j_3}}\cdots\ket{\mathbb{I}_{j_1>j_{k-1}}}\ket{\mathbb{I}_{j_1>j_k}}\\
        \ket{j_2}&\ket{\mathbb{I}_{j_2\geq j_1}}\ket{\mathbb{I}_{j_2>j_3}}\cdots\ket{\mathbb{I}_{j_2>j_{k-1}}}\ket{\mathbb{I}_{j_2>j_k}}\\
        \ket{j_3}&\ket{\mathbb{I}_{j_3\geq j_1}}\ket{\mathbb{I}_{j_3\geq j_2}}\cdots\ket{\mathbb{I}_{j_3>j_{k-1}}}\ket{\mathbb{I}_{j_3>j_k}}\\
        &\cdots\\
        \ket{j_{k-1}}&\ket{\mathbb{I}_{j_{k-1}\geq j_1}}\ket{\mathbb{I}_{j_{k-1}\geq j_2}}\cdots\ket{\mathbb{I}_{j_{k-1}\geq j_{k-2}}}\ket{\mathbb{I}_{j_{k-1}>j_k}}\\
        \ket{j_k}&\ket{\mathbb{I}_{j_k\geq j_1}}\ket{\mathbb{I}_{j_k\geq j_2}}\cdots\ket{\mathbb{I}_{j_k\geq j_{k-2}}}\ket{\mathbb{I}_{j_k\geq j_{k-1}}}
    \end{aligned}\\
    &\xrightarrow{\hspace{1em}2\hspace{1em}}\sum_{j_1,\cdots,j_k}\ket{j_1}\ket{\text{order}_1}\ket{j_2}\ket{\text{order}_2}\cdots\ket{j_k}\ket{\text{order}_k}.
\end{aligned}\end{equation}
Here $\mathbb{I}_E$ is an indicator which takes value $1$ if the expression $E$ is true and $0$ otherwise. The register $\ket{\text{order}_i}$ records the position of $j_i$ when $j_1,\ldots,j_k$ are sorted from small to large. The required operations and the number of ancillary qubits for each step are as follows:

\medskip
\noindent\textbf{Step 1.} We utilize the Greaterthan gate as mentioned in Table~\ref{tab:gates}, each of which requires $O(m^2)$ ancillary qubits for an operation on $m$-qubit integers. To get those like $\ket{\mathbb{I}_{j_2\geq j_1}}$, one can simply copy (equivalent to apply a CNOT gate) the qubit $\ket{\mathbb{I}_{j_1< j_2}}$ and apply an $X$ gate on it. To parallelize this collection of operations, we can first use a quantum fan-out gate to copy each $j_i$ many times. Specifically, we need $(k-1)$ copies of $j_i$ for each $i\in[1,k]$, so the total amount of qubits overhead is $k\cdot(k-1)\cdot\lceil\log_2\eta\rceil=O(k^2\log n)$. The total number of Greaterthan gates used is $\frac{1}{2}k(k-1)$, thus the number of ancillae used to perform these gates is $O(k^2\log^2 n)$. The copied index registers and the work registers of the comparison circuits are then uncomputed, leaving only the comparison-output registers. The total number of qubits required in step 1 is $O(k^2\log^2 n)$.

\medskip
\noindent\textbf{Step 2.} This step is the intermediate measurement layer $M$ mentioned in the circuit~\eqref{eq:circuit}, where we measure all the ancillae to impose a specific order on each $j_i$. To be specific, we mean adding all the measurement outcomes of the ancillae of $j_i$ together. Take $j_2$ for instance, $\text{order}_2=\mathbb{I}_{j_2\geq j_1}+\mathbb{I}_{j_2>j_3}+\cdots+\mathbb{I}_{j_2>j_k}$. Note that the adding operation is done by classical computation after the measurement layer instead of the Hammingweight gates before the measurement layer, and we write $\ket{\text{order}_i}$ here only for the convenience of symbols. It is clear that after this layer of measurement, every $j_i$ will get a unique order, i.e., $\left\{\text{order}_i\right\}_{i=1}^{k}$ is a permutation of the sequence $0,1,\ldots,k-1$ since each $j_i$ is strictly greater or smaller than any of the others in the sense of considering its value and index together (if $j_s$ and $j_t$ are taking the same value at a certain term, then the one with a larger index will get a larger order).

\medskip
Now permute the states to align the ordering ancillae in a progressive order (incurring no quantum-depth overhead), and we get:
\begin{equation}\label{eq:index_reorder}
    \sum_{j_{r_1}\leq\cdots\leq j_{r_k}}\ket{j_{r_1}}\cdots\ket{j_{r_k}},
\end{equation}
where $j_{r_i}$ has order $i$. Here, however, not every weak inequality in Eq.~\eqref{eq:index_reorder} can necessarily be saturated. Consider, for example, the case $k=3$. If the resulting order is $j_1\leq j_2\leq j_3$, corresponding to $\mathrm{order}_1=0$, $\mathrm{order}_2=1$, and $\mathrm{order}_3=2$, then all the displayed equalities are permitted. By contrast, if the order is $j_2\leq j_3\leq j_1$, corresponding to $\mathrm{order}_1=2$, $\mathrm{order}_2=0$, and $\mathrm{order}_3=1$, then the last inequality must in fact be strict, $j_2\leq j_3<j_1$. Indeed, the measurement outcome in this branch includes $\mathbb{I}_{j_1>j_3}=1$, which excludes the possibility $j_1=j_3$.

More generally, an equality between $j_{r_i}$ and $j_{r_{i+1}}$ is forbidden precisely when $r_i>r_{i+1}$. Therefore, the number of inequalities that are already strict is exactly the descent number of the permutation $(r_1,\ldots,r_k)$, which we denote by $\mathrm{des}$. We convert all remaining weak inequalities into strict ones through a lifting operation. For example, for the ordering $j_2\leq j_3<j_1$, we increment both $j_3$ and $j_1$ by one, obtaining $j_2<j_3+1<j_1+1$. In general, for a permutation with $\mathrm{des}=d$, there are $(k-1)-d$ weak inequalities that must be made strict. More explicitly, we add to the $i$-th ordered register the classically determined offset
\begin{equation}
    a_i=\left|\left\{s<i:r_s<r_{s+1}\right\}\right|,\qquad 1\leq i\leq k,
\end{equation}
so that
\begin{equation}
    j_{r_1}+a_1<j_{r_2}+a_2<\cdots<j_{r_k}+a_k.
\end{equation}
Since $0\leq a_i\leq k-1-d$, the largest possible index after lifting is $\eta+k-2-d$. The offsets can be computed classically in polynomial-time as part of $C_2$ in the circuit~\eqref{eq:circuit}. In the quantum circuit, the additions $\ket{j_{r_i}}\longmapsto\ket{j_{r_i}+a_i}$ can be performed in parallel using the $\mathrm{Add}$ gate, with $O(k\log^2 n)$ ancillary qubits. Consequently, conditioned on obtaining a permutation with descent number $d$, the lifting operation produces exactly $\ket{\mathrm{USS}_{\eta+k-1-d,k}}$.

\medskip
For fixed $k$ and $\eta$, the measurement induces a probability distribution over $d=0,1,\ldots,k-1$. Let $p(d)$ denote the probability of obtaining descent number $d$, and let $d(\eta)$ be a value of $d$ that maximizes $p(d)$. The successful copies can be identified from the measurement outcomes using polynomial-time classical processing. Since $d(\eta)$ maximizes a probability distribution supported on at most $k$ values, we have $p\left(d(\eta)\right)\geq1/k$. A more detailed analysis in Appendix~\ref{app:prob}, supported by numerical results, suggests the stronger scaling $1/p(d(\eta))=O(\sqrt{k})$.

It remains to determine the value of $\eta$, which is performed before the quantum circuit is executed. To obtain the target state $\ket{\mathrm{USS}_{n,k}}$, we choose $\eta$ such that $\eta+k-1-d(\eta)=n$. As discussed in Appendix~\ref{app:prob}, a suitable value of $\eta$ can be found by scanning the interval $n-k+1\leq\eta\leq n$ and evaluating the corresponding descent-number probabilities, which requires only polynomial-time classical computation. Combining the preceding steps, the total ancillary-qubit requirement is $O\left(k^2\log^2 n\right)$, while the adaptive quantum depth remains constant. Restoring the normalization and relabeling the ordered index registers, the output is
\begin{equation}
    \frac{1}{\sqrt{\binom{n}{k}}}\sum_{0\leq j_1<\cdots<j_k\leq n-1}\ket{j_1}\cdots\ket{j_k}=\ket{\mathrm{USS}_{n,k}}.
\end{equation}

\end{proof}

To amplify the success probability, we define $\alpha=1/p(d(\eta))$ and perform $R=\lceil c\alpha\rceil$ independent copies of the preparation circuit in parallel, where $c$ is a constant. The probability that none of them yields the desired descent number is bounded by
\begin{equation}
    P_{\mathrm{fail}}=\left(1-\frac{1}{\alpha}\right)^R\leq\exp\left(-\frac{R}{\alpha}\right)\leq e^{-c}.
\end{equation}

The parameter $c$ controls the failure probability: increasing the number of parallel repetitions linearly in $c$ suppresses the failure probability exponentially as $e^{-c}$. This fact is characterized by the following corollary.

\begin{corollary}[Constant-depth preparation of the USS with small failure rate]\label{cor:uss}
    For any integer $0\leq k\leq n$, and for any $c>0$, the state $\ket{\mathrm{USS}_{n,k}}$ can be prepared exactly by a constant-depth adaptive quantum circuit using $O\left(c\alpha k^2\log^2 n\right)$ ancillary qubits and polynomial-time classical processing. The success probability is at least $1-e^{-c}$, where the parameter $\alpha=\alpha(n,k)$ satisfies $\alpha\leq k$.
\end{corollary}

We emphasize that although the bound $\alpha\leq k$ is the one used in the rigorous statement of Corollary~\ref{cor:uss}, the analytical results in limiting regimes and the numerical evidence presented in Appendix~\ref{app:prob} suggest the stronger scaling $\alpha=O(\sqrt{k})$, under which the first term in the ancillary-qubit complexity becomes $O(c k^{5/2}\log^2 n)$.

\medskip
Lemma~\ref{lem:index} provides the key ingredient for our Dicke-state preparation protocol: the resulting USS state serves as an index register specifying the locations of the excitations in the target register. More precisely, for each basis state $\ket{j_1}\cdots\ket{j_k}$ in $\ket{\mathrm{USS}_{n,k}}$, we write $1$ into the $j_1$-th, $\ldots$, $j_k$-th positions of an initially zero target register. This produces the corresponding Hamming-weight-$k$ computational-basis state. Since the indices are arranged in strictly increasing order, we can get the superposition of all those Hamming-weight-$k$ computational-basis states. Moreover, the qubits in the USS state can subsequently be erased. We now describe these two steps.

\medskip
\noindent\textbf{Filling the Target Register.}
The following transformation can be implemented by a constant-depth adaptive quantum circuit using $O(kn\log n\log\log n)$ ancillary qubits:
\begin{equation}\label{eq:filling}
    \ket{\mathrm{USS}_{n,k}}\ket{0}^{\otimes n}\longrightarrow\frac{1}{\sqrt{\binom{n}{k}}}\sum_{0\leq j_1<\cdots<j_k\leq n-1}\ket{j_1}\cdots\ket{j_k}\ket{\bigoplus_{\ell=1}^{k}e_{j_\ell}},
\end{equation}
where $e_j\in\{0,1\}^n$ denotes the one-hot string whose only nonzero entry is at position $j$.

We use the standard construction of Ref.~\cite{buhrman2024state}, which can be decomposed into the following steps:
\begin{equation}\begin{aligned}
    \frac{1}{\sqrt{\binom{n}{k}}}\sum_{j_1<\cdots<j_k}^{n-1}\ket{j_1}\cdots\ket{j_k}\ket{0}^{\otimes n}&\xrightarrow{\hspace{1em}1\hspace{1em}}
    \frac{1}{\sqrt{\binom{n}{k}}}\sum_{j_1<\cdots<j_k}^{n-1}\ket{j_1}\cdots\ket{j_k}\frac{1}{\sqrt{2^n}}\sum_{\ell=0}^{2^n-1}\ket{\ell}\\
    &\xrightarrow{\hspace{1em}2\hspace{1em}}
    \frac{1}{\sqrt{\binom{n}{k}}}\sum_{j_1<\cdots<j_k}^{n-1}\ket{j_1}^{\otimes n}\cdots\ket{j_k}^{\otimes n}\frac{1}{\sqrt{2^n}}\sum_{\ell=0}^{2^n-1}\ket{\ell}^{\otimes k}\\
    &\xrightarrow{\hspace{1em}3\hspace{1em}}
    \frac{1}{\sqrt{\binom{n}{k}}}\sum_{j_1<\cdots<j_k}^{n-1}\ket{j_1}^{\otimes n}\cdots\ket{j_k}^{\otimes n}\frac{1}{\sqrt{2^n}}\sum_{\ell=0}^{2^n-1}(-1)^{\ell_{j_1}+\cdots+\ell_{j_k}}\ket{\ell}^{\otimes k}\\
    &\xrightarrow{\hspace{1em}4\hspace{1em}}
    \frac{1}{\sqrt{\binom{n}{k}}}\sum_{j_1<\cdots<j_k}^{n-1}\ket{j_1}\cdots\ket{j_k}\frac{1}{\sqrt{2^n}}\sum_{\ell=0}^{2^n-1}(-1)^{\ell_{j_1}+\cdots+\ell_{j_k}}\ket{\ell}\\
    &\xrightarrow{\hspace{1em}5\hspace{1em}}
    \frac{1}{\sqrt{\binom{n}{k}}}\sum_{j_1<\cdots<j_k}^{n-1}\ket{j_1}\cdots\ket{j_k}\ket{\bigoplus_{\ell=1}^{k}e_{j_\ell}}.
\end{aligned}\end{equation}

In step~1, transversal Hadamard gates are applied to the target register to prepare the uniform superposition over all $n$-bit strings. In step~2, quantum fan-out gates create the copies of the index and target registers required to parallelize the subsequent operations. This step uses $O(kn\log n)$ ancillary qubits. Step~3 applies, in parallel, phase versions of the $\mathrm{Equal}_i$ gates listed in Table~\ref{tab:gates}. Specifically, instead of flipping the target qubit when $j=i$, the phase version applies a Pauli-$Z$ gate. Such a gate can be obtained from $\mathrm{Equal}_i$ by conjugating its target qubit with Hadamard gates. For every $i\in\{0,\ldots,n-1\}$, we apply the phase version of $\mathrm{Equal}_i$ to the corresponding copy of $\ket{j_1}$ and the qubit $\ket{\ell_i}$. Exactly one of these gates is activated, producing the phase $(-1)^{\ell_{j_1}}$. Repeating this construction in parallel for all $j_1,\ldots,j_k$ produces
$(-1)^{\ell_{j_1}+\cdots+\ell_{j_k}}$. Since each index contains $O(\log n)$ qubits, an $\mathrm{Equal}_i$ gate acting on such an index uses
$O(\log n\log\log n)$ qubits. The $kn$ parallel equality tests in step~3 therefore require $O(kn\log n\log\log n)$ qubits. In step~4, the copies introduced in step~2 are uncomputed by applying the same fan-out operations in reverse. Finally, step~5 applies transversal Hadamard gates to the target register. Every step has constant adaptive depth, and the total ancillary-qubit requirement is $O(kn\log n\log\log n)$.\\

At this point, the target register is in the desired Dicke-state superposition but remains entangled with the index register. It therefore remains to erase the indices coherently.

\medskip
\noindent\textbf{Erasing the Index Register.}
The following transformation can be implemented by a constant-depth adaptive quantum circuit using $O(n^2+kn\log k\log\log k)$ ancillary qubits:
\begin{equation}
\label{eq:clean_index}
    \frac{1}{\sqrt{\binom{n}{k}}}
    \sum_{0\leq j_1<\cdots<j_k\leq n-1}
    \ket{j_1}\cdots\ket{j_k}
    \ket{\bigoplus_{\ell=1}^{k}e_{j_\ell}}
    \longrightarrow
    \frac{1}{\sqrt{\binom{n}{k}}}
    \sum_{0\leq j_1<\cdots<j_k\leq n-1}
    \ket{0}\cdots\ket{0}
    \ket{\bigoplus_{\ell=1}^{k}e_{j_\ell}}.
\end{equation}

An explicit constant-depth implementation of this operation is given in Lemma~4.9 of Ref.~\cite{buhrman2024state}. However, in their most resource-intensive part, $n$ Hammingweight gates on prefixes of the target register of lengths 
$1,2,\ldots,n$ need to be performed in parallel. According to Table~\ref{tab:gates}, an exact Hammingweight gate acting on $m$ input qubits can be implemented using $O(m^2)$ ancillary qubits, and the total ancillary-qubit requirement is therefore $\sum_{m=1}^{n}O(m^2)=O(n^3)$. In the following, we show that the resource estimation can be improved using the trick of shared dyadic block counts.\\

For each computational-basis component of the target register, we write $x=\bigoplus_{\ell=1}^{k}e_{j_\ell}=(x_0,\ldots,x_{n-1})$ for short, so that $x_t=1$ when $t\in\{j_1,\ldots,j_k\}$ and $0$ otherwise. Since the indices are strictly increasing, $j_\ell$ is the position of the $\ell$-th nonzero bit of $x$. It can therefore be recovered from the prefix ranks
\begin{equation}\label{eq:prefix-rank}
    R_t=\sum_{u=0}^{t-1}x_u,\qquad0\leq t\leq n,
\end{equation}
where we set $R_0=0$. In particular,
\begin{equation}\label{eq:index-from-prefix-rank}
    t=j_\ell\quad\Longleftrightarrow\quad x_t=1\ \text{and}\ R_t=\ell-1.
\end{equation}

Let $L=\lceil\log_2 n\rceil$, $N=2^L$ and append $(N-n)$ qubits initialized in $\ket{0}$ to the target register. For every $0\leq s\leq L$ and $0\leq q<N/2^s$, define the dyadic interval $I_{s,q}=\left\{q2^s,q2^s+1,\ldots,(q+1)2^s-1\right\}$
and its Hamming weight $w_{s,q}=\sum_{t\in I_{s,q}}x_t$.

We first create required copies of the target qubits using quantum
fan-out, after which compute all the dyadic weights in parallel. At scale $s$, there are $N/2^s$ intervals, each containing $2^s$ bits. According to Table~\ref{tab:gates}, the exact Hammingweight gate on a $2^s$-qubit register requires $O(2^{2s})$ qubits. The total number of qubits required to compute all dyadic weights is therefore
\begin{equation}\label{eq:dyadic-counting-cost}
    \sum_{s=0}^{L}\frac{N}{2^s}\cdot O\left(2^{2s}\right)=O\left(N\sum_{s=0}^{L}2^s\right)=O(N^2)=O(n^2).
\end{equation}

For every $t$, the prefix interval $\{0,\ldots,t-1\}$ has a canonical decomposition into at most $L=O(\log n)$ disjoint dyadic intervals. Let $\mathcal{D}_t$ denote this decomposition. Eq.~\eqref{eq:prefix-rank} can then be written as $R_t=\sum_{I_{s,q}\in\mathcal{D}_t}w_{s,q}$.

Thus, each prefix rank is obtained by adding at most $O(\log n)$ integers, each represented using $O(\log k)$ bits. Exact iterated addition of this form can be implemented by a polynomial-size constant-depth threshold circuit~\cite{HESSE2002695threshold}. By the exact collapse $\mathrm{QNC}_f^0=\mathrm{QTC}_f^0$ established in Ref.~\cite{takahashi2016collapse}, it can be implemented exactly by a constant-depth quantum circuit with unbounded fan-out using $\mathrm{polylog}(n)$ qubits for each value of $t$. All $n$ prefix ranks can therefore be computed in parallel using $n\,\mathrm{polylog}(n)$ additional qubits. The dyadic-weight registers can be copied to the corresponding addition circuits using $O(n\log n\log k)$ qubits, since each prefix uses at most $O(\log n)$ such registers.

We next convert the prefix ranks into one-hot encodings of the ordered indices. For every $1\leq \ell\leq k$ and $0\leq t\leq n-1$, define $g_{\ell,t}=x_t\,\mathbb{I}_{R_t=\ell-1}$. These bits can be computed in parallel using the $\mathrm{Equal}_{\ell-1}$ gates in Table~\ref{tab:gates}, followed by Toffoli gates controlled by $x_t$. Equation~\eqref{eq:index-from-prefix-rank} implies that $(g_{\ell,0},\ldots,g_{\ell,n-1})=e_{j_\ell}$.

There are $nk$ equality tests, each acting on an $O(\log k)$-qubit prefix-rank register. This step therefore requires $O\left(kn\log k\log\log k\right)$ qubits.

It remains to use these one-hot registers to erase the original binary indices. Taking $j_1$ as an example, we perform the following operations:
\begin{equation}\begin{aligned}
    \ket{j_1}\ket{e_{j_1}}&\xrightarrow{\hspace{1em}1\hspace{1em}}\sum_{i=0}^{N-1}(-1)^{i\cdot j_1}\ket{i}\ket{e_{j_1}}\\
    &\xrightarrow{\hspace{1em}2\hspace{1em}}\sum_{i=0}^{N-1}(-1)^{i\cdot j_1}\ket{i}^{\otimes n}\ket{e_{j_1}}\\
    &\xrightarrow{\hspace{1em}3\hspace{1em}}\sum_{i=0}^{N-1}\ket{i}\ket{e_{j_1}}\\
    &\xrightarrow{\hspace{1em}4\hspace{1em}}\ket{0}\ket{e_{j_1}},
\end{aligned}\end{equation}
where we ignore some coefficients and $i\cdot j_1$ denotes the inner product of $i$ and $j_1$ as bit strings. Step~1 applies transversal Hadamard gates on $\ket{j_1}$. Step~2 copies $\ket{i}$ for $n$ times using fan-out. In step~3, for each $0\le\ell\le n-1$, we apply a bunch of Pauli-$Z$ gates on the $\ell$-th copy of $\ket{i}$ whose positions together form the binary representation of $\ell$ controlled by the $\ell$-th bit of $\ket{e_{j_1}}$ to eliminate the phase factor $(-1)^{i\cdot j_1}$. This can be done by conjugating fan-out with Hadamard gates. Step~4 discards $\sum_{i=0}^{N-1}\ket{i}$ since it is now disentangled with $\ket{e_{j_1}}$. The whole process uses $O(n\log n)$ ancillary qubits, and since there are $k$ indices to eliminate, the total number of qubits required is $O(kn\log n)$.

Finally, we uncompute the registers $g_{\ell,t}$, the prefix ranks $R_t$, and the dyadic weights $w_{s,q}$ in reverse order. All these registers were computed solely from the target register, which remains unchanged, so the uncomputation restores every ancillary qubit to $\ket{0}$. Combining all the stuff above, the index-erasing operation can be accomplished in constant depth, and the total number of ancillary qubits used by the index-erasure operation is $O(n^2+kn\log k\log\log k)$. Note that $n\,\mathrm{polylog}(n)$ and $O(kn\log n)$ can always be bounded by $O(n^2+kn\log k\log\log k)$.

\medskip
We are now ready to prove Theorem~\ref{thm:Dicke}.

\begin{proof}
    The complete Dicke-state preparation protocol is
    \begin{equation}\begin{aligned}
        \ket{0}&\xrightarrow{\hspace{1em}1\hspace{1em}}
        \ket{\mathrm{USS}_{n,k}}\ket{0}^{\otimes n}\\
        &\xrightarrow{\hspace{1em}2\hspace{1em}}
        \frac{1}{\sqrt{\binom{n}{k}}}\sum_{0\le j_1<\cdots<j_k\le n-1}\ket{j_1}\cdots\ket{j_k}\ket{\bigoplus_{l=1}^{k}e_{j_l}}\\
        &\xrightarrow{\hspace{1em}3\hspace{1em}}
        \frac{1}{\sqrt{\binom{n}{k}}}\sum_{0\le j_1<\cdots<j_k\le n-1}\ket{0}\cdots\ket{0}\ket{\bigoplus_{l=1}^{k}e_{j_l}}=\ket{0}\ket{D_k^n}.
    \end{aligned}\end{equation}
    
    Step~1 uses Lemma~\ref{lem:index} to prepare the uniform subset superposition introduced in Definition~\ref{def:uss}. Step~2 writes the excitations specified by the index registers into the target register using Eq.~\eqref{eq:filling}, and step~3 erases the index registers using Eq.~\eqref{eq:clean_index}.

    The first step uses $O\left(k^2\log^2 n\right)$ ancillary qubits, while the latter two steps use $O(kn\log n\log\log n)$ and $O(n^2+kn\log k\log\log k)$ ancillary qubits, respectively. Hence the total ancillary-qubit requirement is
    \begin{equation}
        O\left(n^2+k^2\log^2 n+kn\log n\log\log n\right)
    \end{equation}
    
    All quantum steps have constant adaptive depth. The polynomial-time classical computation and the success probability at least $1/k$ follow from Lemma~\ref{lem:index}. Conditioned on success, every transformation is exact, and the target register is therefore prepared exactly in the state $\ket{D_k^n}$.
\end{proof}

If we replace Lemma~\ref{lem:index} with Corollary~\ref{cor:uss} to suppress the failure rate, we can have the following corollary.
\begin{corollary}[Constant-Depth Dicke-State Preparation with small failure rate]\label{cor:dicke}
    For any $n$ and $0\leq k\leq n$, and for any $c>0$, the Dicke-$(n,k)$ state can be exactly prepared using $O(1)$ depth quantum computation with $O\left(n^2+c\alpha k^2\log^2 n+kn\log n\log\log n\right)$ ancillary qubits and polynomial-time classical computation. The success probability is at least $1-e^{-c}$, where the parameter $\alpha=\alpha(n,k)$ satisfies $\alpha\leq k$.
\end{corollary}

Again, a more detailed analysis and numerical results in Appendix~\ref{app:prob} suggests the stronger scaling $\alpha=O(\sqrt{k})$.

\section{From Dicke-State Preparation to Arbitrary Symmetric States}\label{sec:sym}

In this section, we establish a general framework that reduces the preparation of arbitrary symmetric states to unitary Dicke-state preparation. More precisely, we show that any family of unitary circuits for preparing Dicke states can be lifted to a circuit for preparing arbitrary superpositions of Dicke states, with only a polynomial overhead in ancillary qubits. We first outline the high-level structure of the construction.

Let $E$ denote a one-hot index register and $T$ an $n$-qubit target register. The preparation of an arbitrary symmetric state can be decomposed into the following three steps:
\begin{equation}\begin{aligned}
    \ket{0}_E\ket{0}^{\otimes n}_T&\xrightarrow{\hspace{1em}1\hspace{1em}}\sum_{k=0}^{n}a_k\ket{e_k}_E\ket{0}^{\otimes n}_T\\
    &\xrightarrow{\hspace{1em}2\hspace{1em}}\sum_{k=0}^{n}a_k\ket{e_k}_E\ket{D_k^n}_T\\
    &\xrightarrow{\hspace{1em}3\hspace{1em}}\ket{0}_E\sum_{k=0}^{n}a_k\ket{D_k^n}_T,
\end{aligned}\end{equation}
where $\ket{e_k}$ denotes the one-hot encoding of $k$ in an $(n+1)$-qubit register. In step~1, the desired coefficients $\{a_k\}_{k=0}^{n}$ are coherently encoded into the one-hot index register. In step~2, the value stored in this register coherently controls the preparation of the corresponding Dicke state. Finally, step~3 erases the index register, leaving the desired symmetric state in the target register.

To implement the first step, we use the following result from Ref.~\cite{zi2025constant}.

\begin{lemma}\label{lem:any_state}
    \textnormal{\cite{zi2025constant}} An arbitrary $2^n$-qubit quantum state of the form $\sum_{j=0}^{2^n-1}\alpha_j e^{i\theta_j}\ket{e_j}$ can be prepared from the initial state $\ket{0}^{\otimes 2^n}$ using a constant-depth quantum circuit. The circuit has a size of $O(n4^n)$ and utilizes $O(n4^n)$ ancillae.
\end{lemma}

Lemma~\ref{lem:any_state} shows that an arbitrary superposition over a one-hot basis can be prepared in constant depth, although the required number of ancillary qubits is exponential in the logarithm of the dimension of the one-hot register. In the present setting, however, the symmetric subspace has dimension only $n+1$. We may therefore set $m=\left\lceil\log_2(n+1)\right\rceil$ and apply Lemma~\ref{lem:any_state} with the amplitudes corresponding to $j>n$ set to zero. The state $\sum_{k=0}^{n}a_k\ket{e_k}$ can thus be prepared in constant depth using $O\left(m4^m\right)=O\left(n^2\log n\right)$ ancillary qubits. The unused one-hot positions remain unoccupied and may be omitted from the subsequent construction.

We can now state the main result of this section.

\begin{theorem}[From Dicke states to arbitrary symmetric states]
\label{thm:symm}
    Suppose that, for any integer $0\leq k\leq n$, there exists a unitary circuit $U_{n,k}$ composed of bounded-fan-in gates and unbounded quantum fan-out gates preparing the Dicke state $\ket{D_k^n}$ exactly with each using at most $O(\xi)$ ancillary qubits, i.e., $U_{n,k}\left(\ket{0}^{\otimes n}\ket{0}^{\otimes \xi}\right)=\ket{D_k^n}\ket{0}^{\otimes \xi}$. Suppose further that each $U_{n,k}$ has depth at most $O(L)$. Then any $n$-qubit symmetric state $\sum_{k=0}^{n}a_k\ket{D_k^n}$ can be prepared exactly by an adaptive circuit of depth $O(L)$ using $O\left(n^2\log n+n\xi\right)$ ancillary qubits.
\end{theorem}

Theorem~\ref{thm:symm} shows that unitary Dicke-state preparation is sufficient, up to the stated polynomial overhead, for preparing arbitrary symmetric states. Since the Dicke states $\left\{\ket{D_k^n}\right\}_{k=0}^{n}$ form an orthonormal basis of the symmetric subspace, it is enough to prepare the desired amplitudes coherently in an index register and then conditionally apply the corresponding Dicke-state preparation circuits. The theorem therefore provides a generic lifting procedure that converts any family of clean unitary Dicke-state preparation circuits into a preparation circuit for arbitrary symmetric states.

An immediate consequence is that improvements in unitary Dicke-state preparation translate directly into improved protocols for general symmetric-state preparation. In this sense, apart from the overhead required to coherently combine the different Hamming-weight sectors, the complexity of arbitrary symmetric-state preparation is governed by the complexity of the underlying Dicke-state preparation circuits.

Now we give a formal proof of Theorem~\ref{thm:symm}.

\begin{proof}
The first step, $\ket{0}\longmapsto\sum_{k=0}^{n}a_k\ket{e_k}$, can be implemented in constant depth using $O(n^2\log n)$ ancillary qubits, as established by Lemma~\ref{lem:any_state}. To implement the second step coherently, we first make the following observation.

\medskip
\noindent\textbf{From a unitary to its controlled version.}
Suppose that a unitary circuit $U$ acts on $N=O(n+\xi)$ qubits, consists of bounded-fan-in gates and unbounded quantum fan-out gates, and has depth $O(L)$. Then its controlled version $C$-$U$ can be implemented in depth $O(L)$ using $O(n+\xi)$ additional ancillary qubits.

Indeed, $U$ can be decomposed into $O(L)$ layers of gates with disjoint supports. In each layer, the number of gates is at most $O(N)$. We may therefore create $O(N)$ coherent copies of the control qubit using quantum fan-out and use a distinct copy to control each gate in the layer. The controlled version of a bounded-fan-in gate remains a bounded-fan-in gate. Moreover, a controlled fan-out gate can also be implemented in constant depth. Explicitly, the transformation
\begin{equation}
    \ket{c}\ket{x}\ket{y_1,\ldots,y_m}\longmapsto\ket{c}\ket{x}\ket{y_1\oplus cx,\ldots,y_m\oplus cx}
\end{equation}
can be realized by first computing $cx$ into an ancillary qubit using a Toffoli gate, applying a fan-out gate controlled by this ancillary qubit, and then uncomputing it. The control copies can be erased after the controlled circuit has been applied. Hence the depth remains $O(L)$ and the additional width is $O(n+\xi)$.

\medskip
Using this observation, we introduce $n+1$ mutually disjoint $n$-qubit registers $A_0,A_1,\ldots,A_n$ and apply the controlled circuits $C$-$U_{n,k}$ in parallel, with the $k$-th qubit of the one-hot register $E$ serving as the control of $U_{n,k}$. This implements
\begin{equation}\begin{aligned}\label{eq:parallel-controlled-dicke}
    \sum_{k=0}^{n}a_k\ket{e_k}_E\bigotimes_{j=0}^{n}\ket{0}^{\otimes n}_{A_j}\longmapsto
    \sum_{k=0}^{n}a_k\ket{e_k}_E\left(\bigotimes_{j=0}^{k-1}\ket{0}^{\otimes n}_{A_j}\right)\ket{D_k^n}_{A_k}\left(\bigotimes_{j=k+1}^{n}\ket{0}^{\otimes n}_{A_j}\right).
\end{aligned}\end{equation}

The controlled circuits have total depth $O(L)$ and use $O\left(n(\xi+n)\right)=O(n\xi+n^2)$ ancillary qubits, including the $n+1$ registers $A_j$.

\medskip
\noindent\textbf{Transfer Dicke states into the target register.}
We next introduce an $n$-qubit target register $T$, initialized in $\ket{0}^{\otimes n}$, and coherently transfer the selected Dicke state from $A_k$ to $T$:
\begin{equation}\begin{aligned}\label{eq:swap_op}
    &\sum_{k=0}^{n}a_k\ket{e_k}_I\ket{0}^{\otimes n}_T\left(\bigotimes_{j=0}^{k-1}\ket{0}^{\otimes n}_{A_j}\right)\ket{D_k^n}_{A_k}\left(\bigotimes_{j=k+1}^{n}\ket{0}^{\otimes n}_{A_j}\right)\longmapsto
    \sum_{k=0}^{n}a_k\ket{e_k}_I\ket{D_k^n}_T\bigotimes_{j=0}^{n}\ket{0}^{\otimes n}_{A_j}.
\end{aligned}\end{equation}

This transformation is obtained by applying a SWAP between $T$ and $A_j$ controlled by the $j$-th qubit of $E$ for every $j$. Since $E$ is supported entirely on the one-hot subspace, exactly one controlled-SWAP is activated in each computational-basis branch.

We now show that all these controlled-SWAP operations can be implemented in constant depth. For $q\in\{1,\ldots,n\}$, let $T_q$ and $A_{j,q}$ denote the $q$-th qubits of $T$ and $A_j$, respectively, and let $E_j$ denote the $j$-th qubit of $E$. Define
\begin{equation}\begin{aligned}
    V_q&=\prod_{j=0}^{n}\mathrm{CCNOT}(E_j,T_q;A_{j,q}),\\
    W_q&=\prod_{j=0}^{n}\mathrm{CCNOT}(E_j,A_{j,q};T_q),
\end{aligned}\end{equation}
where $\mathrm{CCNOT}(a,b;c)$ denotes a Toffoli gate with controls $a,b$ and target $c$. On the one-hot subspace of $E$, the product of controlled-SWAP gates acting on the $q$-th qubits satisfies
\begin{equation}\label{eq:parallel-cswap-decomposition}
    \prod_{j=0}^{n}\mathrm{CSWAP}(E_j;T_q,A_{j,q})=V_qW_qV_q.
\end{equation}

Indeed, in every basis branch of the one-hot register, all factors except the one corresponding to the nonzero position of $E$ act trivially, and Eq.~\eqref{eq:parallel-cswap-decomposition} reduces to the standard three-Toffoli decomposition of a controlled-SWAP gate. The middle layer $W_q$ can be converted into the same form as $V_q$ by Hadamard conjugation. Let $\mathcal{H}_q=H_{T_q}\prod_{j=0}^{n}H_{A_{j,q}}$. Using the symmetry of the controlled-$Z$ representation of the Toffoli gate, we have
\begin{equation}\label{eq:toffoli-conjugation}
    W_q=\mathcal{H}_q V_q\mathcal{H}_q.
\end{equation}

Consequently, the controlled-SWAP transformation for the $q$-th qubits can be implemented using three layers of gates of the form $V_q$, together with a constant number of Hadamard layers.

To perform each $V_q$ in constant depth, we coherently copy $T_q$ into $n+1$ ancillary qubits using fan-out and apply all the corresponding Toffoli gates in parallel. When all values of $q$ are treated simultaneously, we also copy each qubit $E_j$ sufficiently many times so that it can control the operations on all $n$ target positions. The total number of copies required is $O(n^2)$. These copies are uncomputed after each parallel Toffoli layer and may be reused in the subsequent layers. It follows that all controlled-SWAP operations in Eq.~\eqref{eq:swap_op} can be implemented in constant depth using $O(n^2)$ ancillary qubits. This completes the second step of the high-level construction.

\medskip
It remains to erase the one-hot register $E$. Let $L=\left\lceil\log_2(n+1)\right\rceil$ and introduce an $L$-qubit register $B$. Applying the exact Hammingweight gate to the target register gives
\begin{equation}\label{eq:compute-weight-for-cleaning}
    \sum_{k=0}^{n}a_k\ket{e_k}_E\ket{D_k^n}_T\ket{0}_B
    \longmapsto
    \sum_{k=0}^{n}a_k\ket{e_k}_E\ket{D_k^n}_T\ket{k}_B.
\end{equation}

We then use quantum fan-out to create $n+1$ coherent copies $B_0,\ldots,B_n$ of the binary Hamming-weight register. For each $j\in\{0,\ldots,n\}$, we apply an $\mathrm{Equal}_j$ gate to $B_j$, using the qubit $E_j$ as its target. This implements
\begin{equation}
    \ket{k}_{B_j}\ket{(e_k)_j}_{E_j}
    \longmapsto\ket{k}_{B_j}\ket{(e_k)_j\oplus\mathbb{I}_{k=j}}_{E_j}
    =\ket{k}_{B_j}\ket{0}_{E_j}.
\end{equation}

Applying these equality tests in parallel therefore erases the entire one-hot register:
\begin{equation}
    \sum_{k=0}^{n}a_k\ket{e_k}_E\ket{D_k^n}_T\bigotimes_{j=0}^{n}\ket{k}_{B_j}
    \longmapsto
    \ket{0}^{\otimes(n+1)}_E\sum_{k=0}^{n}a_k\ket{D_k^n}_T\bigotimes_{j=0}^{n}\ket{k}_{B_j}.
\end{equation}

Finally, we undo the fan-out operation and apply the inverse Hammingweight circuit to erase $B$. The exact Hamming-weight computation uses $O(n^2)$ ancillary qubits. Copying its $O(\log n)$-qubit output $n+1$ times requires $O(n\log n)$ qubits, while the parallel equality tests require $O(n\log n\log\log n)$ qubits. Thus, the entire erasure step has constant depth and uses $O(n^2)$ ancillary qubits.

Combining all three steps, the total depth is $O(L)$, and the total number of ancillary qubits is $O(n^2\log n)+O(n\xi+n^2)+O(n^2)=O(n^2\log n+n\xi)$. The final state of the target register is exactly $\sum_{k=0}^{n}a_k\ket{D_k^n}$, which completes the proof.
\end{proof}

Theorem~\ref{thm:symm} provides a general procedure for lifting a family of unitary Dicke-state preparation circuits to a preparation circuit for arbitrary symmetric states. But unfortunately, Our protocol in Section~\ref{sec:dicke} cannot be inserted directly into this framework. Theorem~\ref{thm:symm} requires a clean unitary circuit $U_{n,k}$ that can be applied coherently under the control of the one-hot register. By contrast, the protocol in Theorem~\ref{thm:Dicke} relies on intermediate measurements, classical processing, and post-selection of successful branches. If such a protocol were applied conditionally in a coherent superposition over $k$, its measurement records could become correlated with the Hamming-weight label and thereby destroy the coherence between different Dicke sectors. Thus, the adaptive protocol of Theorem~\ref{thm:Dicke} does not directly satisfy the assumptions of Theorem~\ref{thm:symm}.

Therefore, we use a clean unitary circuit that exactly prepares $\ket{D_k^n}$ for every $0\leq k\leq n$ constructed recently~\cite{joshi2026constantdepthunitarypreparationdicke}. It has a constant depth and uses $O\left(n^2\sqrt{\log n}\right)$ ancillary qubits. Substituting this result into Theorem~\ref{thm:symm} immediately yields the following corollary.

\begin{corollary}[Constant-depth symmetric-state preparation]\label{cor:symmetric}
    For any integer $n\geq 1$ and any complex coefficients $\{a_k\}_{k=0}^{n}$ satisfying $\sum_{k=0}^{n}|a_k|^2=1$, the symmetric state $\sum_{k=0}^{n}a_k\ket{D_k^n}$ can be prepared exactly by a constant-depth quantum circuit using $O\left(n^3\sqrt{\log n}\right)$ ancillary qubits.
\end{corollary}

To the best of our knowledge, no constant-depth exact preparation protocol for arbitrary symmetric states with a better asymptotic ancillary-qubit complexity is currently known.

\section{Discussions}\label{sec:con_dis}

In this work, we have developed a constant-depth adaptive protocol for preparing arbitrary Dicke-$(n,k)$ states through the uniform subset superposition $\ket{\mathrm{USS}_{n,k}}$. A single execution succeeds with probability at least $1/k$ and, conditioned on success, prepares the target state exactly using $O\left(n^2+k^2\log^2 n+kn\log n\log\log n\right)$ ancillary qubits. The failure probability can be suppressed exponentially through independent repetitions without increasing the adaptive quantum depth. We have also established a general lifting framework that converts clean unitary Dicke-state preparation circuits into preparation circuits for arbitrary symmetric states. Combined with recent constant-depth Dicke-state constructions, this gives an exact constant-depth preparation protocol for arbitrary $n$-qubit symmetric states using $O(n^3\sqrt{\log n})$ ancillary qubits.

Several questions remain concerning the optimal resource complexity of the protocol. First, our construction uses $O\left(n^2+k^2\log^2 n+kn\log n\log\log n\right)$ ancillary qubits to prepared Dicke states, while the lower bound of required ancillary qubits is $\Omega(n)$~\cite{liu2025lowdepthquantumsymmetrization}. Whether we can find a better protocol using less ancillary qubits or improve the lower bound of the required ancillary qubits remains open. Second, it is natural to ask whether $\ket{\mathrm{USS}_{n,k}}$ admits a deterministic clean unitary preparation in constant depth with as less ancillary qubits as possible. Such a construction would eliminate the ordering step and, more importantly, allow the USS preparation circuit to be used coherently as a controlled subroutine. Third, although the dyadic-block construction reduces the ancillary complexity of index erasure substantially, it remains open whether the transformation in Eq.~\eqref{eq:clean_index} can be implemented with nearly linear width. This question is closely related to the complexity of exact constant-depth Hamming-weight computation. Complementary lower bounds for constant-depth preparation of the USS and Dicke states would help determine whether the remaining overhead is intrinsic.

The USS may also be useful as an algorithmic primitive beyond Dicke-state preparation. An important next step is to determine whether the ordered-subset representation also permits efficient implementations of the Johnson-graph update primitives used in quantum-walk algorithms~\cite{ambainis2007quantumwalk,childs2005quantumalgorithms,magniez2007triangle}, as well as in nested quantum-walk constructions~\cite{Jeffery2013walk} and the quantum algorithms for topological data analysis~\cite{McArdle2026streamlined}. Such implementations would turn the USS from a state-preparation primitive into a building block for complete combinatorial quantum algorithms.

The lifting framework of Theorem~\ref{thm:symm} suggests a more general coherent-combination principle: whenever an orthonormal family of states is indexed by a register of polynomial dimension and each state admits a clean preparation circuit, superpositions over the family may be constructed by coherently multiplexing these circuits. Extending this perspective to other permutation-invariant sectors, Schur-basis states, or symmetry-resolved many-body states would be worthwhile.

Finally, the practical relevance of the protocol depends on its compilation into concrete experimental architectures. The construction is naturally suited to platforms supporting mid-circuit measurements, classical feedforward, and highly parallel or long-range entangling operations. Its probabilistic nature also permits different implementation strategies. One may repeat the preparation sequentially before the Dicke state enters a coherent computation, increasing the wall-clock preparation time without increasing the peak qubit count, or prepare several copies in parallel and route a successful output onward, trading additional width for reduced latency. A realistic resource analysis should therefore account not only for quantum depth, but also for measurement and feedforward latency, connectivity, qubit reset and reuse, and classical-control bandwidth. It will also be important to study how measurement errors and imperfect comparison or counting operations affect the output fidelity, and whether efficient verification or purification procedures can be incorporated without losing the low-depth advantage. Investigating architecture-specific and fault-tolerant implementations may clarify when adaptive ordering provides a practical advantage over fully unitary Dicke-state preparation.

\bibliographystyle{apsrev}

\bibliography{./tex/bibSymPrep.bib}

\appendix

\section{Analysis of $p(d)$, $\alpha$ and $\eta$ in Lemma~\ref{lem:index} and Corollary~\ref{cor:uss}}\label{app:prob}
In this appendix, we analyze the $p(d)$, $\alpha$ and $\eta$ concerned in the preparation of the USS states in Lemma~\ref{lem:index} and Corollary~\ref{cor:uss} in detail. Our setting is that: for fixed $k$ and $\eta$, the probability of getting descent number $d$ after the measurement layer is $p(d)$, $d(\eta)$ is the most probable descent number, and $p(d(\eta))=1/\alpha$. Firstly, we perform an explicit calculation of $p(d)$, in which we have to introduce the Eulerian number $A(k,d)$ (the number of $d$-descent sequences among all $k!$ permutations of $\{0,1,\ldots,k-1\}$).

\begin{lemma}
    In our current setting, the probability of getting descent number $d$ after the measurement layer in Lemma~\ref{lem:index} for fixed $\eta$ and $k$ is:
    \begin{equation}\label{eq:pd}
        p(d)=\frac{1}{\eta^k}\cdot\binom{\eta+k-1-d}{k}\cdot A(k,d).
    \end{equation}
\end{lemma}

\begin{proof}
    In $\sum_{j_1,\cdots,j_k=0}^{\eta-1}\ket{j_1}\cdots\ket{j_k}$, we refer to each instance of $\ket{j_1}\cdots\ket{j_k}$ as a J-register. Clearly there are $\eta^k$ of them. After getting an order in Eq.~\eqref{eq:get_order} with the help of Greaterthan gates (we refer to each instance of the order behind the J-register as an order register), the entire state can be rearranged as
    \begin{equation}
        \begin{aligned}
            &\left(\sum\ket{\text{J-register}}\right)\ket{\text{order-register with $\mathrm{des}=0$}}+\cdots+\left(\sum\ket{\text{J-register}}\right)\ket{\text{order-register with $\mathrm{des}=0$}}\\
            +&\left(\sum\ket{\text{J-register}}\right)\ket{\text{order-register with $\mathrm{des}=1$}}+\cdots+\left(\sum\ket{\text{J-register}}\right)\ket{\text{order-register with $\mathrm{des}=1$}}\\
            +&\cdots\\
            +&\left(\sum\ket{\text{J-register}}\right)\ket{\text{order-register with $\mathrm{des}=k-1$}}+\cdots+\left(\sum\ket{\text{J-register}}\right)\ket{\text{order-register with $\mathrm{des}=k-1$}},
        \end{aligned}
    \end{equation}
        
    in which if several J-registers share the same order register, we collect and put them in a summation $\sum\ket{\text{J-register}}$, which is a ``potential'' USS state (a USS before the lifting procedure). Note that the $\mathrm{des}=d$ here does not mean the sequence in the order register has descent number $d$, but the sequence of the indices in the J-register after reordering according to the order register (the sequence of $r_1,\ldots,r_k$ in Eq.~\eqref{eq:index_reorder}) has descent number $d$. Since there is a bijection between the order register and the order of the indices, this way of writing is acceptable, and the total number of the ``potential'' USS states with descent number $d$ is $A(k,d)$. After lifting, the upper bound of $j$ in such potential USS is lifted to $\eta+k-2-d$, so there are $\binom{\eta+k-1-d}{k}$ J-registers in total in each of the potential USS states. Thus, the probability of getting a potential USS with descent number $d$ is characterized by Eq.~\eqref{eq:pd}.
\end{proof}

To verify the normalization condition of $p(d)$, note that we have the Worpitzky equality~\cite{Worpitzky1883}
\begin{equation}
    \sum_{d=0}^{k-1}A(k,d)\cdot\binom{m+d}{k}=m^k.
\end{equation}

Also, since we know the sequence of Eulerian numbers is palindromic, meaning $A(k,d)=A(k,k-1-d)$, we have $\sum_{d=0}^{k-1}p(d)=1$.

\medskip
For the remainder of this appendix, we will focus on analyzing some basic properties and extreme cases of $p(d)$ as well as presenting some numerical results instead of analyzing $p(d)$ thoroughly, which is beyond the main purpose of this work. We can show that there is exactly one maximum point for $p(d)$.
\begin{lemma}
    For fixed $\eta$ and $k$, the function $p(d)$ defined in Eq.~\eqref{eq:pd} is unimodal.
\end{lemma}
\begin{proof}
    Consider the following quantity
    \begin{equation}
        R(d)=\frac{p(d+1)}{p(d)}=\frac{A(k,d+1)}{A(k,d)}\cdot\left(1-\frac{k}{\eta+k-1-d}\right).
    \end{equation}
    When $R(d)>1$, we have $p(d+1)>p(d)$, corresponding to the increasing interval of $p(d)$, while the case $R(d)<1$ corresponds to the decreasing interval of $p(d)$. We know that the sequence of $A(k,d)$ is log-concave~\cite{Brenti1989UnimodalLA}, meaning that
    \begin{equation}
        A(k,d)^2\ge A(k,d-1)A(k,d+1).
    \end{equation}
    Combined with the fact that $\left(1-\frac{k}{\eta+k-1-d}\right)$ is monotonically decreasing with respect to $d$, we have that $R(d)$ is monotonically decreasing. As a result, $p(d)$ is unimodal.
\end{proof}

Now we discuss the magnitude of $\alpha$ given $\eta$ and $k$. The lemma below deals with the case when $k\ll\eta$ (which is equivalent to $k\ll n$).
\begin{lemma}\label{lem:eta_gg_k}
    For fixed $\eta$ and $k$ with $k\ll \eta$, $d(\eta)$ is roughly $\lfloor\frac{k-1}{2}\rfloor$ and $\alpha$ is of the order $O(\sqrt{k})$. As a result, $\eta$ can be roughly taken as $n-\lceil\frac{k-1}{2}\rceil$ in this case.
\end{lemma}
\begin{proof}
    In the situation $k\ll\eta$, we have:
    \begin{equation}\begin{aligned}
        p(d)&=\frac{1}{\eta^k}\cdot\frac{(\eta+k-1-d)!}{(\eta-1-d)!}\cdot\frac{A(k,d)}{k!}\\
        &\approx\frac{1}{\eta^k e^k}\cdot\frac{(\eta+k-1-d)^{\eta+k-1-d}}{(\eta-1-d)^{\eta-1-d}}\cdot\frac{A(k,d)}{k!}\\
        &=\frac{1}{e^k}\cdot\left(1+\frac{k-1-d}{\eta}\right)^k\cdot\left(1+\frac{k}{\eta-1-d}\right)^{\eta-1-d}\cdot\frac{A(k,d)}{k!}\\
        &\approx\frac{A(k,d)}{k!},
    \end{aligned}\end{equation}
    where the Stirling's formula has been used. As we have mentioned before, the sequence of Eulerian numbers is palindromic and its maximum is taken at $\mathrm{des}=\lfloor\frac{k-1}{2}\rfloor$, so $d(\eta)$ in this case is roughly $\lfloor\frac{k-1}{2}\rfloor$ and $\eta$ is roughly $n-\lceil\frac{k-1}{2}\rceil$. To calculate the magnitude of $\alpha$, note that we have the following asymptotic formula~\cite{carlitz1972asymptotic}:
    \begin{equation}
        \frac{1}{k!}A(k,[x_k])=\sqrt{\frac{6}{\pi(k+1)}}\cdot e^{-\frac{x^2}{2}}+O(k^{-\frac{3}{4}}),
    \end{equation}
    where $x_k=\sqrt{\frac{1}{12}(k+1)}\cdot x+\frac{1}{2}(k+1)$. Let $x=0$, we have $p(d(\eta))\approx A(k,\frac{k+1}{2})/k!\approx\sqrt{\frac{6}{\pi(k+1)}}$. Thus we can assert $\alpha=O(\sqrt{k})$.
\end{proof}

Another extreme case when $\eta\ll k$ is rather simple. Taking $\eta=1$, we have $p(d)=\binom{k-d}{k}\cdot A(k,d)$, which takes the value 0 except for $d=0$, where $p(0)=1$, leading to $\alpha=1$. Now we have established a basic expectation about the image of $p(d)$ with $\eta$ in different ranges: let $k$ be fixed and $\eta$ increase from 1 to $\infty$, the peak of $p(d)$ should gradually move from $d=0$ to $d=\lfloor\frac{k-1}{2}\rfloor$, and $\alpha$ should increase from 1 to $O(\sqrt{k})$. Although we believe this fact can be proven by some detailed mathematical analysis, such effort is beyond the purpose of this paper. Instead, we present some numerical results in Fig.~\ref{fig:pd_eta} which support our conjecture.

\begin{figure}[htbp]
    \centering
    \captionsetup{justification=justified, singlelinecheck=false}
    \begin{subfigure}{0.32\textwidth}
        \centering
        \includegraphics[width=\linewidth]{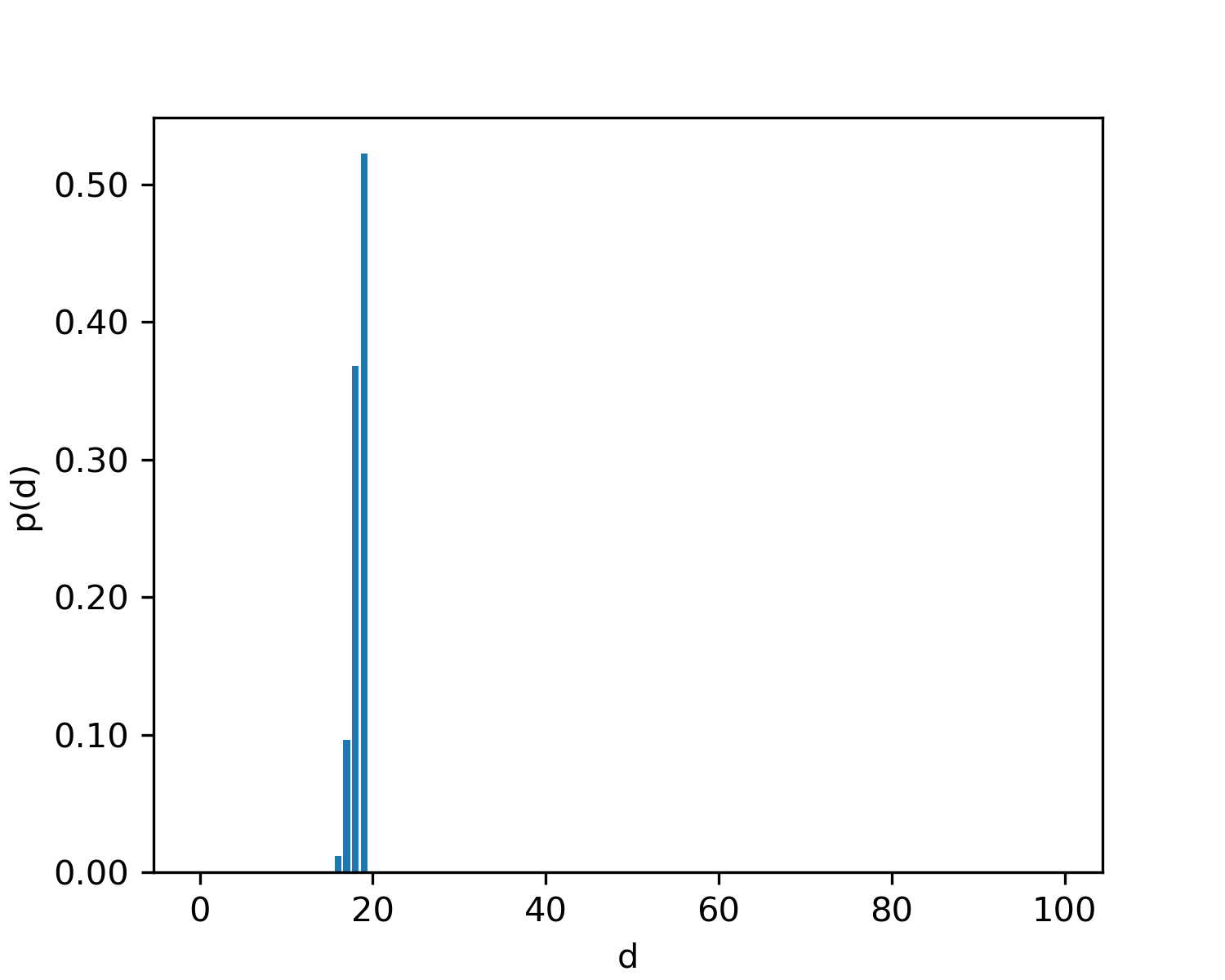}
        \caption{$\eta=20$}
    \end{subfigure}
    \hfill
    \begin{subfigure}{0.32\textwidth}
        \centering
        \includegraphics[width=\linewidth]{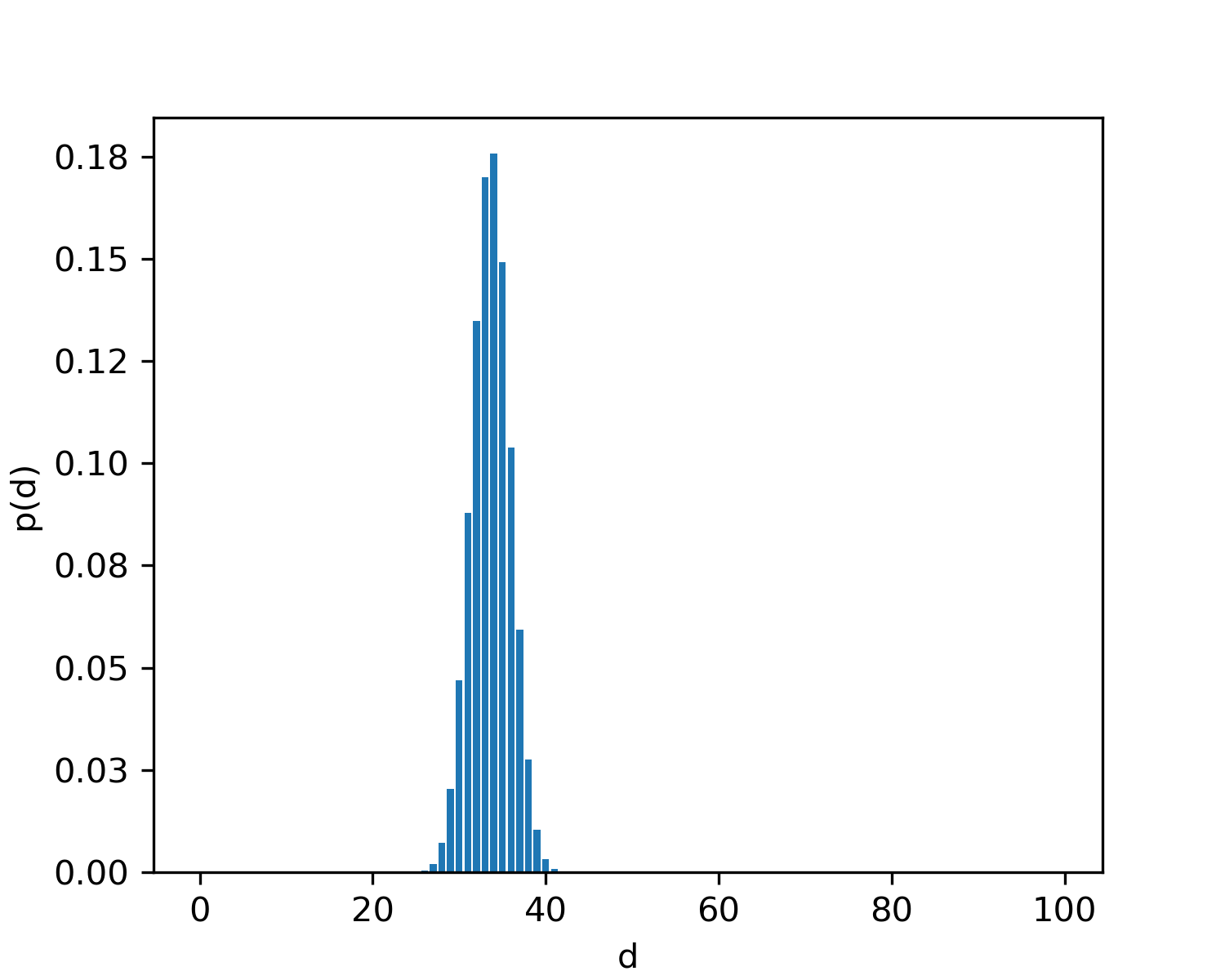}
        \caption{$\eta=50$}
    \end{subfigure}
    \hfill
    \begin{subfigure}{0.32\textwidth}
        \centering
        \includegraphics[width=\linewidth]{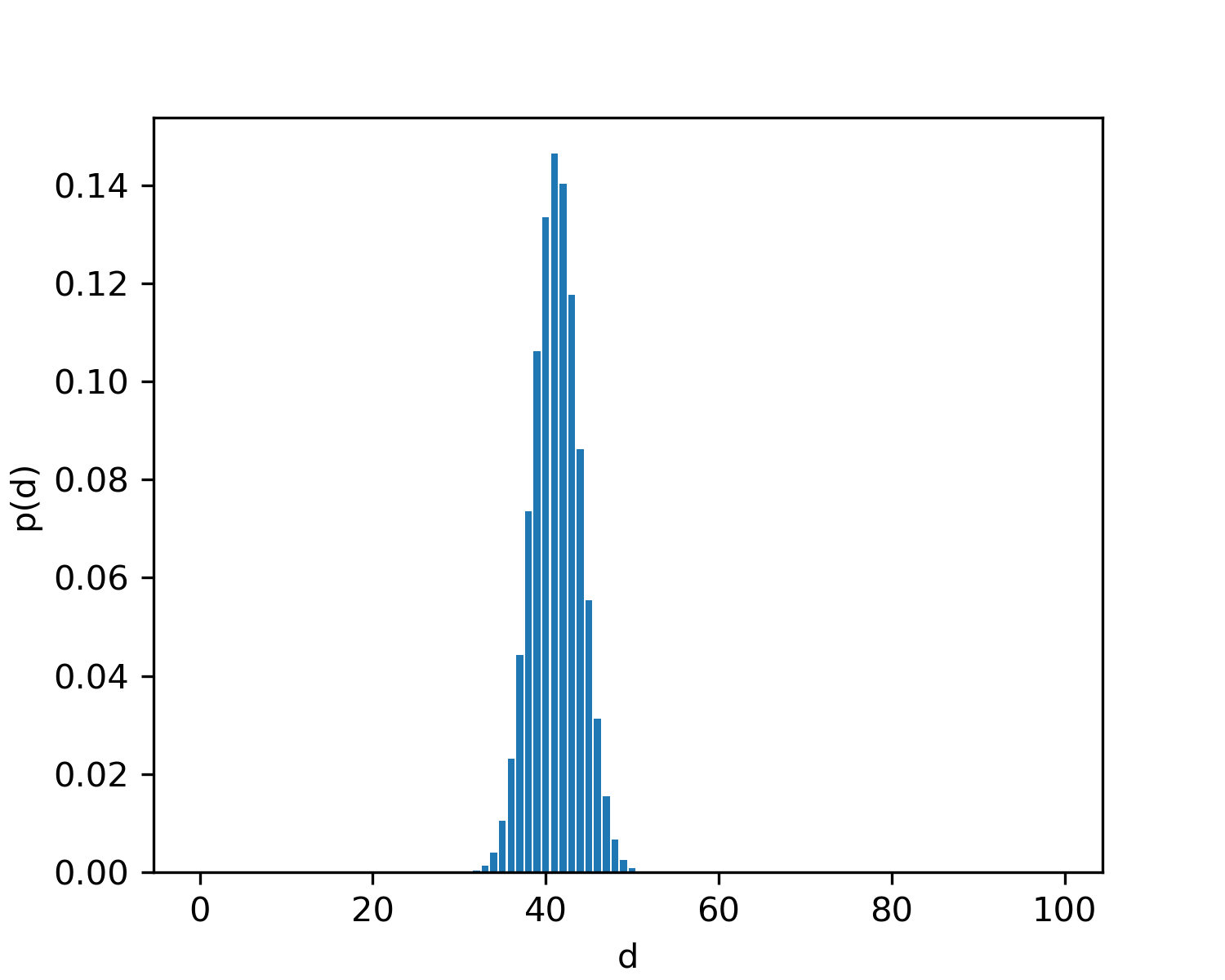}
        \caption{$\eta=100$}
    \end{subfigure}

    \vspace{0.5cm}

    \begin{subfigure}{0.32\textwidth}
        \centering
        \includegraphics[width=\linewidth]{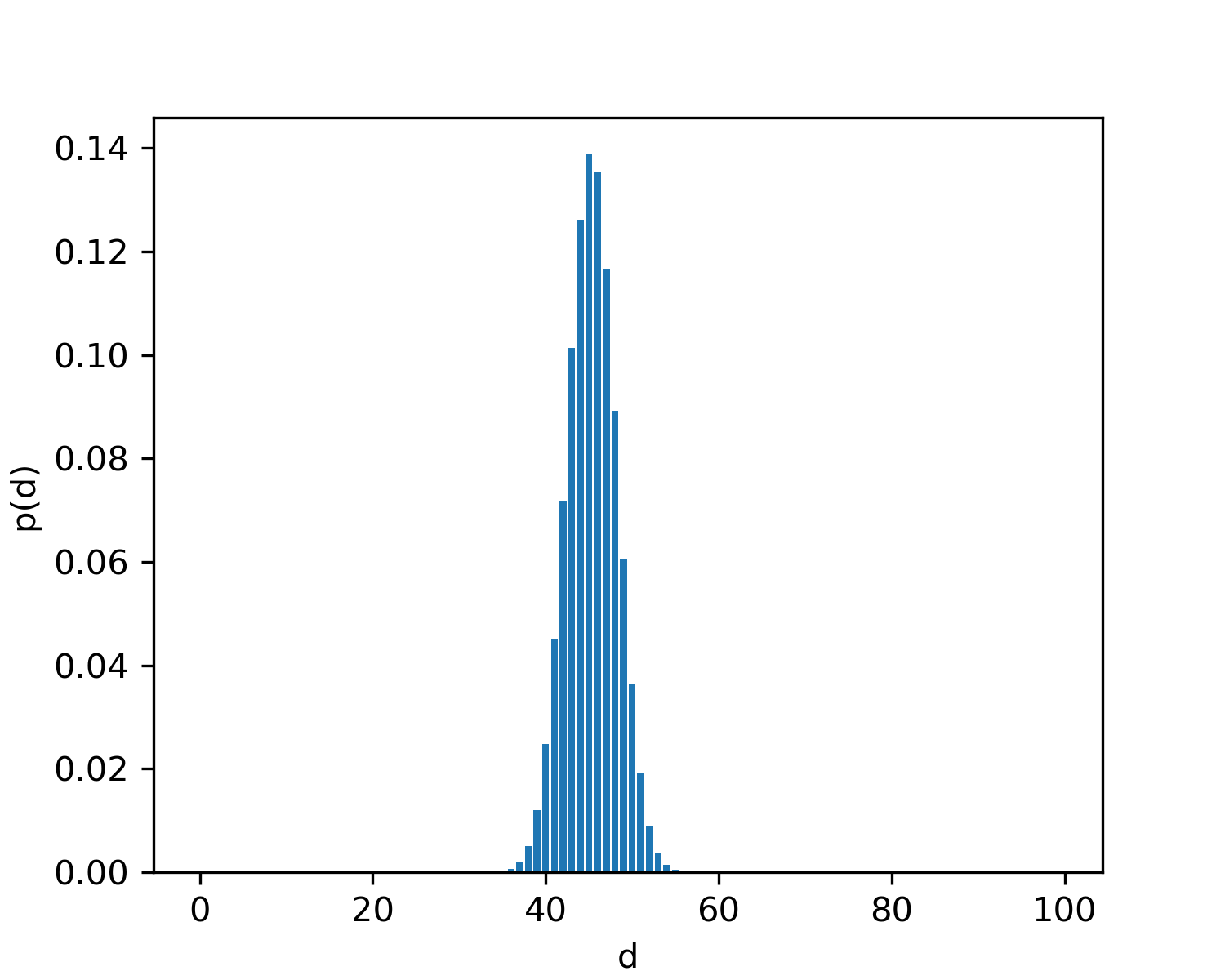}
        \caption{$\eta=200$}
    \end{subfigure}
    \hfill
    \begin{subfigure}{0.32\textwidth}
        \centering
        \includegraphics[width=\linewidth]{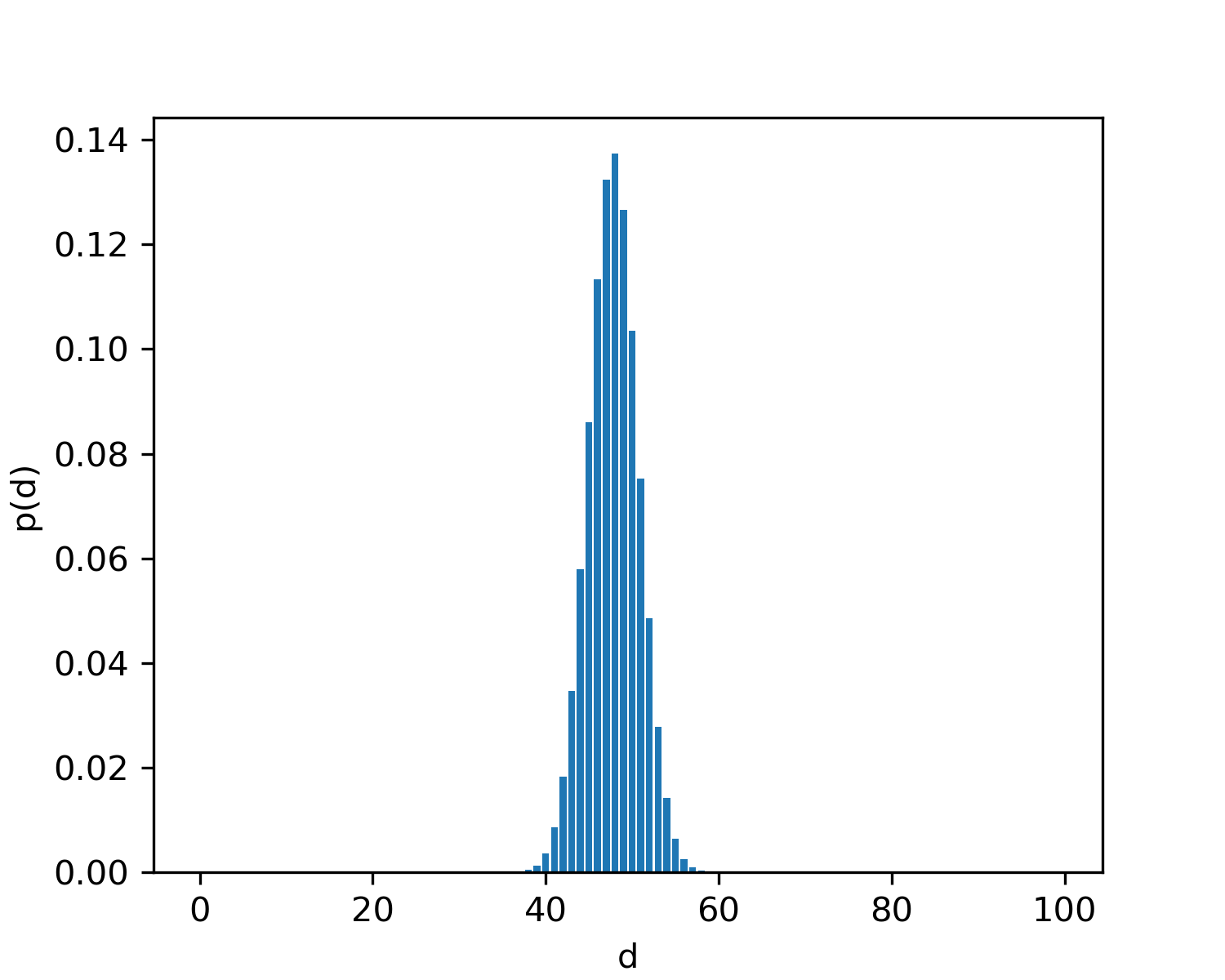}
        \caption{$\eta=500$}
    \end{subfigure}
    \hfill
    \begin{subfigure}{0.32\textwidth}
        \centering
        \includegraphics[width=\linewidth]{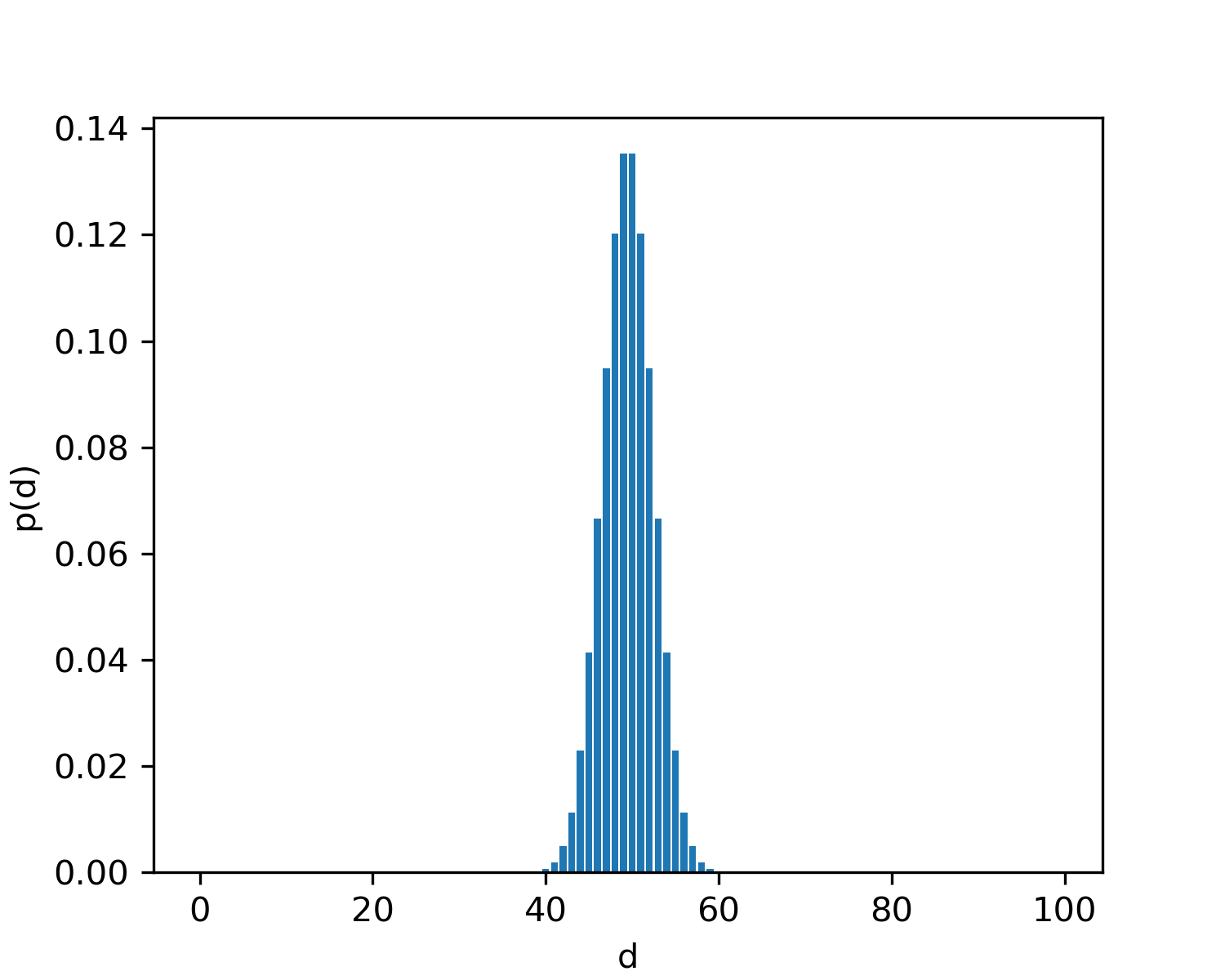}
        \caption{$\eta=\infty$}
    \end{subfigure}
    \caption{Fix $k=100$, the images of $p(d)$ with different values of $\eta$. As $\eta$ increases from 20 to $\infty$ ($\eta=\infty$ is drawn by $p(d)=\frac{A(k,d)}{k!}$), the peak gradually moves from left to right and eventually stops at $\frac{k}{2}$. The value of $p(d(\eta))$ also decreases to roughly $\sqrt{\frac{6}{\pi k}}\approx0.138$, which is consistent with our analysis in Lemma~\ref{lem:eta_gg_k}.}
    \label{fig:pd_eta}
\end{figure}

Lastly, we give a brief comment on how to select $\eta$ according to $n$ and $k$ at the beginning of Lemma~\ref{lem:index} or Corollary~\ref{cor:uss} using polynomial-time classical computation. The method is quite straightforward: under the constraint $\eta+k-1-d(\eta)=n$, we know $\eta=n-k+1+d(\eta)\in[n-k+1,n]$, so all we have to do is to scan $\eta$ from $n-k+1$ to $n$, calculate $d(\eta)$ and see whether $\eta+k-1-d(\eta)$ matches $n$. Note that since the peak of $p(d)$ has a certain width, some suboptimal solutions of $d\approx d(\eta)$ is acceptable if there happens to be no exact match. Also, since we have the following recursion formula for the Eulerian number:
\begin{equation}
    A(k,d)=(k-d)A(k-1,d-1)+(d+1)A(k-1,d),
\end{equation}
the total computational load is acceptable. In practice, we have the following two trains of thought to design the algorithm:
\begin{itemize}
    \item Maximize $p(d)$: for each $\eta\in[n-k+1,n]$, calculate $\bar{d}=\eta+k-1-n$ and $p(\bar{d})$. Select $\bar{\eta}$ and the corresponding $\bar{d}$ with the largest $p(\bar{d})$.
    \item Minimize $\eta$ and time consumption: set a threshold for $p(d)$ ($p_\text{th}=\frac{1}{\sqrt{k}}$ for instance). For $\eta\in[n-k+1,n]$, calculate $\tilde{d}=\eta+k-1-n$ and $p(\tilde{d})$. Once $p(\tilde{d})>p_\text{th}$, immediately return $\tilde{\eta}$ and the corresponding $\tilde{d}$.
\end{itemize}

Again, some numerical results are given in Table~\ref{tab:n_k_eta}.

\begin{table}[t]
\centering
\caption{The selection of $\eta$ given $n$ and $k$.}
\label{tab:n_k_eta}

\renewcommand{\arraystretch}{1.25}
\setlength{\tabcolsep}{12pt}

\resizebox{\linewidth}{!}{
\begin{tabular}{c|cccc|c|cccc}
\hline
$n$ & 150 & 150 & 200 & 500 & $n$ & 150 & 150 & 200 & 500 \\
$k$ & 50 & 100 & 100 & 100 & $k$ & 50 & 100 & 100 & 100 \\
\hline
$\bar{\eta}$ & 124 & 91 & 145 & 449 & $\tilde{\eta}$ & 123 & 89 & 143 & 447 \\
$\bar{d}$ & 23 & 40 & 44 & 48 & $\tilde{d}$ & 22 & 38 & 42 & 46 \\
$p(\bar{d})$ & 0.194 & 0.147 & 0.141 & 0.137 & $p(\tilde{d})$ & 0.182 & 0.106 & 0.120 & 0.118 \\
$\bar{\alpha}=1/p(\bar{d})$ & 5.16 & 6.79 & 7.09 & 7.32 & $\tilde{\alpha}=1/p(\tilde{d})$ & 5.49 & 9.41 & 8.34 & 8.47 \\
$d(\bar{\eta})$ & 23 & 40 & 44 & 48 & $d(\tilde{\eta})$ & 23 & 40 & 44 & 48 \\
$p(d(\bar{\eta}))$ & 0.194 & 0.147 & 0.141 & 0.137 & $p(d(\tilde{\eta}))$ & 0.194 & 0.149 & 0.141 & 0.137 \\
\hline
\end{tabular}
}

\vspace{4pt}
\begin{minipage}{\textwidth}
\footnotesize
\textit{Notes.}
In the left side of this table we calculate $\bar{\eta}$ and the corresponding $\bar{d}$ by maximizing $p(d)$, where $d(\bar{\eta})$ and $p(d(\bar{\eta}))$ are the optimal descent number and the corresponding maximum probability given $\bar{\eta}$. In the right side of this table we calculate $\tilde{\eta}$ and the corresponding $\tilde{d}$ by minimizing $\eta$ and time consumption. Note that given $n=150$, there is no essential difference between $k=50$ and $k=100$ (they refer to the same Dicke state up to transversal Pauli $X$). Although for $k=50$, its $p(\bar{d})$ ($p(\tilde{d})$) is larger, the $\bar{\eta}$ ($\tilde{\eta}$) corresponding to $k=100$ is notably smaller. This suggests a potential trade-off in choosing $k$ or $n-k$ in the protocol of Lemma~\ref{lem:index}.
\end{minipage}

\end{table}

\end{document}